\documentclass[final,12pt]{elsarticle}

\makeatletter
\def\ps@pprintTitle{%
 \let\@oddhead\@empty
 \let\@evenhead\@empty
 \def\@oddfoot{}%
 \let\@evenfoot\@oddfoot}
\makeatother

\usepackage{tikz}
\usepackage{hyperref}
\usepackage{amsfonts}
\usepackage{epsfig}
\usepackage{makeidx}
\usepackage{amsmath,amssymb,amssymb,latexsym}
\usepackage{enumerate}
\usepackage{graphicx}
\usepackage{mathtools}

\newcommand{\vertex}[4][black]{
    \draw[#1, fill=#1, inner sep=0pt] (#2, #3) circle (0.075) node(#4){};
}

\newcommand{\edge}[3][]{
    \draw[#1] (#2) -- (#3);
}

 \DeclareMathOperator{\thin}{thin}
\DeclareMathOperator{\pthin}{pthin}
\DeclareMathOperator{\indthin}{thin_{ind}}
\DeclareMathOperator{\indf}{\textit{f}_{ind}}
\DeclareMathOperator{\indpthin}{pthin_{ind}}
\DeclareMathOperator{\compthin}{thin_{cmp}}

\DeclareMathOperator{\lip}{lip} \DeclareMathOperator{\diam}{diam}
\DeclareMathOperator{\mim}{mim}

\DeclareMathOperator{\comppthin}{pthin_{cmp}}
\DeclareMathOperator{\opthin}{(p)thin}
\DeclareMathOperator{\opindthin}{(p)thin_{ind}}
\DeclareMathOperator{\opcompthin}{(p)thin_{cmp}}

\newtheorem{theorem}{Theorem}

\newtheorem{lemma}[theorem]{Lemma}
\newtheorem{corollary}[theorem]{Corollary}

\newtheorem{remark}{Remark}

\newenvironment{proof}{\textit{Proof.}}{\hfill $\ \square \ $ \\}

\date{}

\begin{document}

\begin{frontmatter}

\title{Thinness of product graphs}

\author[UBA,ICC]{Flavia Bonomo-Braberman}
\ead{fbonomo@dc.uba.ar}

\author[ICC]{Carolina L. Gonzalez}
\ead{cgonzalez@dc.uba.ar}

\author[UERJ]{Fabiano S. Oliveira}
\ead{fabiano.oliveira@ime.uerj.br}

\author[UFRJ]{Moys\'es S. Sampaio Jr.}
\ead{moysessj@cos.ufrj.br}

\author[UERJ,UFRJ]{Jayme L. Szwarcfiter}
\ead{jayme@nce.ufrj.br}

\address[UBA]{Universidad de Buenos Aires. Facultad de Ciencias Exactas y Naturales. Departamento de Computaci\'on. Buenos Aires,
Argentina.}
\address[ICC]{CONICET-Universidad de Buenos Aires. Instituto de
Investigaci\'on en Ciencias de la Computaci\'on (ICC). Buenos
Aires, Argentina.}
\address[UERJ]{Universidade do Estado do Rio de Janeiro, Brazil.}
\address[UFRJ]{Universidade Federal do Rio de Janeiro, Brazil.}

\begin{abstract} The thinness of a graph is a width parameter that generalizes some properties of interval graphs, which are exactly the graphs of thinness one. Many NP-complete problems can be solved in polynomial time for graphs with bounded thinness, given a suitable representation of the graph.
In this paper we study the thinness and its variations of graph
products. We show that the thinness behaves ``well'' in general
for products, in the sense that for most of the graph products
defined in the literature, the thinness of the product of two
graphs is bounded by a function (typically product or sum) of
their thinness, or of the thinness of one of them and the size of
the other. We also show for some cases the non-existence of such a
function.
\end{abstract}

\begin{keyword} graph operations, product graphs, proper thinness, thinness, independent thinness, complete thinness. \end{keyword}






\end{frontmatter}

\section{Introduction}\label{sec:intro}

A graph $G=(V,E)$ is \emph{$k$-thin} if there exist an ordering
$v_1, \dots , v_n$ of $V$ and a partition of $V$ into $k$ classes
$(V^1,\dots,V^k)$ such that, for each triple $(r,s,t)$ with
$r<s<t$, if $v_r$, $v_s$ belong to the same class and $v_t v_r \in
E$, then $v_t v_s \in E$. The minimum $k$ such that $G$ is
$k$-thin is called the \emph{thinness} of $G$ and denoted by
$\thin(G)$.

The thinness is unbounded on the class of all graphs, and graphs
with bounded thinness were introduced in~\cite{M-O-R-C-thinness}
as a generalization of interval graphs.

In~\cite{B-D-thinness}, the concept of \emph{proper thinness} is
defined in order to obtain an analogous generalization of proper
interval graphs, and it is proved that the proper thinness is
unbounded on the class of interval graphs. A graph $G=(V,E)$ is
\emph{proper $k$-thin} if there exist an ordering $v_1, \dots ,
v_n$ of $V$ and a partition of $V$ into $k$ classes
$(V^1,\dots,V^k)$ such that, for each triple $(r,s,t)$ with
$r<s<t$, if $v_r$, $v_s$ belong to the same class and $v_t v_r \in
E$, then $v_t v_s \in E$, and if $v_s$, $v_t$ belong to the same
class and $v_t v_r \in E$, then $v_s v_r \in E$. The minimum $k$
such that $G$ is proper $k$-thin is called the \emph{proper
thinness} of $G$ and denoted by $\pthin(G)$.

{The parameters of thinness  and proper thinness represent how far
a graph is from being an interval and proper interval graph,
respectively. The class of (proper) $1$-thin graphs is that of
(proper) interval graphs. This is so because, considering a
$1$-partitioning, a (strongly) consistent ordering is sufficient
to characterize (proper) interval
graphs~\cite{Ola-interval,Rob-uig}.}

When a representation of the graph as a $k$-thin graph is given,
for a constant value $k$, a wide family of NP-complete problems
can be solved in polynomial time, and the family of problems can
be enlarged when a representation of the graph as a proper
$k$-thin graph is given, for a constant value
$k$~\cite{B-D-thinness,B-M-O-thin-tcs,M-O-R-C-thinness}. Hardness
results for relaxed versions of this family of problems are shown
in~\cite{Bentz-thin}, even for classes of graphs with thinness one
or two.

Many operations are defined over graphs, and some of them arise in
structural characterizations of particular graph families or graph
classes, for example union and join in the case of
cographs~\cite{CorneilLerchsStewart81}. The operations known as
\emph{graph products} are those for which the graph obtained by
the operation of graphs $G_1$ and $G_2$ has as vertex set the
Cartesian product $V(G_1) \times V(G_2)$. Different graph products
are determined by the rules that define the edge set of the
obtained graph. The main properties of these products are surveyed
in~\cite{I-K-prod}.

There is a wide literature about the behavior of graph parameters
under graph operations, and in particular graph products. For the
chromatic number, it includes the famous conjecture of Hedetniemi
(1966)~\cite{Hed-conj}, that remained open more than fifty years,
it was shown to hold for many particular classes, and was recently
disproved by Shitov~\cite{Shitov2019}. Other results on the
chromatic number and its variations in product graphs can be found
in~\cite{A-H-L-prod,Mat-product,B-K-T-V-lagos15-dam,C-G-H-L-M-prod,C-F-prod,G-S-lex,Har,H-M-prod,I-K-prod,I-K-R-Cart,J-P-prod,J-T-color,
Kla-col-prod,K-P-b-col,K-M-prod1,K-M-prod2,Sabidussi1964,V-V-col,Zhu-frac},
and on domination in product graphs
in~\cite{H-R-dom-prod,Hed-thesis,I-K-prod,I-K-R-Cart}. For width
parameters, there is a recent paper studying the boxicity and
cubicity of product graphs~\cite{C-I-M-R-box-prod}.

In~\cite{B-D-thinness}, the behavior of the thinness and proper
thinness under the graph operations union, join, and Cartesian
product is studied. These results allow, respectively, to fully
characterize $k$-thin graphs by forbidden induced subgraphs within
the class of cographs, and to show the polynomiality of the
$t$-rainbow domination problem for fixed $t$ on graphs with
bounded thinness.

In this paper, we give bounds for the thinness and proper thinness
of union and join of graphs, as well as the thinness and proper
thinness of the lexicographical, Cartesian, direct, strong,
disjunctive, modular, homomorphic and hom-products of graphs in
terms of invariants of the component graphs. We also show that in
some cases such bounds do not exist. Furthermore, we describe new
general lower and upper bounds for the thinness of graphs. Also,
we consider the concepts of independent and complete (proper)
thinness, corresponding to the situations in which the classes are
all independent or complete sets. Several of the results on the
bounds of products of graphs are given additionally for these
cases.

The organization of the paper is as follows. In
Section~\ref{sec:thin} we state the main definitions and present
some basic results on thinness and proper thinness. In
Section~\ref{sec:families}, we determine the (proper) thinness of
some graph families, and prove some lower and upper bounds for the
parameters. Section~\ref{sec:thin-and-oper} contains the main
results of the paper, namely bounds of (proper) thinness for
different binary operations, in terms of the (proper) thinness of
their factors. Some concluding remarks form the last section.

\section{Definitions and basic results}\label{sec:thin}

All graphs in this work are finite, undirected, and have no loops
or multiple edges. For all graph-theoretic notions and notation
not defined here, we refer to West~\cite{West}. Let $G$ be a
graph. Denote by $V(G)$ its vertex set, by $E(G)$ its edge set, by
$\overline G$ its complement, by $\Delta(G)$ (resp. $\delta(G)$)
the maximum (resp. minimum) degree of a vertex in $G$. A graph is
\emph{$k$-regular} if every vertex has degree $k$.

Denote by $N(v)$ the neighborhood of a vertex $v$ in $G$, and by
$N[v]$ the closed neighborhood $N(v)\cup\{v\}$. If $X \subseteq
V(G)$, denote by $N(X)$ the set of vertices of $G$ having at least
one neighbor in $X$. A vertex $v$ of $G$ is \emph{universal}
(resp. \emph{isolated}) if $N[v] = V(G)$ (resp. $N(v) =
\emptyset$).

{An \emph{homogeneous set}} is a proper subset $X \subset V(G)$ of
at least two vertices such that every vertex not in $X$ is
adjacent either to all the vertices in $X$ or to none of them.

Denote by $G[W]$ the subgraph of $G$ induced by $W\subseteq V(G)$,
and by $G - W$ or $G \setminus W$ the graph $G[V(G)\setminus W]$.
A subgraph $H$ (not necessarily induced) of $G$ is a
\emph{spanning subgraph} if $V(H)=V(G)$.

Denote the size of a set $S$ by $|S|$. A \emph{clique} or
\emph{complete set} (resp.\ \emph{stable set} or \emph{independent
set}) is a set of pairwise adjacent (resp.\ nonadjacent) vertices.
We use \emph{maximum} to mean maximum-sized, whereas
\emph{maximal} means inclusion-wise maximal. The use of
\emph{minimum} and \emph{minimal} is analogous. The size of a
maximum clique (resp.\ stable set) in a graph $G$ is denoted by
$\omega(G)$ (resp.\ $\alpha(G)$).

A \emph{vertex cover} is a set $S$ of vertices of a graph $G$ such
that each edge of $G$ has at least one endpoint in $S$. Denote by
$\tau(G)$ the size of a minimum vertex cover in a graph $G$.

A graph is called \emph{trivial} if it has only one vertex. A
graph is \emph{complete} if its vertices are pairwise adjacent.
Denote by $K_n$ the complete graph of size~$n$.

Let $H$ be a graph and $t$ a natural number. The disjoint union of
$t$ disjoint copies of the graph $H$ is denoted by $tH$. In
particular, $\overline{tK_2}$ is the complement of a matching of
size~$t$. Denote by $\mim(G)$ the size of a maximum induced
matching of a graph $G$.

{Denote by $P_n$ the path on $n$ vertices.} Given a connected
graph $G$, let $\lip(G)$ be the length of the longest induced path
of $G$, and $\diam(G)$ its diameter. A graph is a \emph{cograph}
if it contains no induced $P_4$.

For a positive integer $r$, the \emph{$(r \times r)$-grid} $GR_r$
is the graph whose vertex set is $\{(i,j) : 1 \leq i, j \leq r\}$
and whose edge set is $\{(i,j)(k,l) : |i - k| + |j - l| = 1,
\mbox{ where } 1 \leq i,j,k, l \leq r \}$.

The \emph{crown graph} $CR_n$ {(also known as Hiraguchi graph)} is
the graph on $2n$ vertices obtained from a complete bipartite
graph $K_{n,n}$ by removing a perfect matching.

A \emph{dominating set} in a graph is a set of vertices such that
each vertex outside the set has at least one neighbor in the set.

A \emph{coloring} of a graph is an assignment of colors to its
vertices such that any two adjacent vertices are assigned
different colors. The smallest number $t$ such that $G$ admits a
coloring with $t$ colors (a \emph{$t$-coloring}) is called the
\emph{chromatic number} of $G$ and is denoted by $\chi(G)$. A
coloring defines a partition of the vertices of the graph into
stable sets, called \emph{color classes}.

A graph $G(V,E)$ is a \emph{comparability graph} if there exists
an ordering $v_1, \dots , v_n$ of $V$ such that, for each triple
$(r,s,t)$ with $r<s<t$, if $v_r v_s$ and $v_s v_t$ are edges of
$G$, then so is $v_r v_t$.\ Such an ordering is a
\emph{comparability ordering}. A graph is a \emph{co-comparability
graph} if its complement is a comparability graph.

In the context of thinness, an ordering $v_1, \dots , v_n$ of
$V(G)$ and a partition of $V(G)$ satisfying that for each triple
$(r, s, t)$ with $r <s<t$, if $v_r$, $v_s$ belong to the same
class and $v_tv_r \in E(G)$, then $v_tv_s \in E(G)$, are said to
be \emph{consistent}. If both $v_1, \dots , v_n$ and $v_n, \dots ,
v_1$ are consistent with the partition, the partition and the
ordering $v_1,\dots,v_n$ are said to be \emph{strongly
consistent}. Notice that a graph is (proper) $k$-thin if and only
if it admits a vertex ordering and a vertex partition into $k$
classes that are (strongly) consistent.

We will often use the following definitions and results.

Let $G$ be a graph and ${<}$ an ordering of its vertices. The
graph $G_{<}$ has $V(G)$ as vertex set, and $E(G_{<})$ is such
that for $v < w$ in the ordering, $vw \in E(G_{<})$ if and only if
there is a vertex $z$ in $G$ such that $w < z$ in the ordering,
$zv \in E(G)$ and $zw \not \in E(G)$. Similarly, the graph
$\tilde{G}_{<}$ has $V(G)$ as vertex set, and $E(\tilde{G}_{<})$
is such that for $v < w$ in the ordering, $vw \in
E(\tilde{G}_{<})$ if and only if either $vw \in E(G_{<})$ or there
is a vertex $x$ in $G$ such that $x < v$ in the ordering, $xw \in
E(G)$ and $xv \not \in E(G)$. An edge of $G_{<}$ (respectively
$\tilde{G}_{<}$) represents that its endpoints cannot belong to
the same class in a vertex partition that is consistent
(respectively strongly consistent) with the ordering ${<}$.

%

\begin{theorem}\label{thm:thin-comp-order}\cite{B-D-thinness,B-M-O-thin-tcs} Given a graph $G$ and an ordering ${<}$ of its vertices,
the graphs $G_{<}$ and $\tilde{G}_{<}$ have the following
properties:
    \begin{enumerate}[(1)]
\item\label{item:2} the chromatic number of $G_{<}$ (resp.
$\tilde{G}_{<}$) is equal to the minimum integer $k$ such that
there is a partition of $V(G)$ into $k$ sets that is consistent
(resp. strongly consistent) with the order ${<}$, and the color
classes of a valid coloring of $G_{<}$ (resp. $\tilde{G}_{<}$)
form a partition consistent (resp. strongly consistent) with
${<}$;

\item $G_{<}$ and $\tilde{G}_{<}$ are co-comparability graphs;

\item if $G$ is a co-comparability graph and ${<}$ a comparability
ordering of $\overline{G}$, then $G_{<}$ and $\tilde{G}_{<}$ are
spanning subgraphs of $G$.
\end{enumerate}
\end{theorem}

Since co-comparability graphs are perfect~\cite{Meyn-co-comp},
$\chi(G_{<})=\omega(G_{<})$ and
$\chi(\tilde{G}_{<})=\omega(\tilde{G}_{<})$. We thus have the
following.

\begin{corollary}\label{cor:thin-comp-order} Let $G$ be a graph, and $k$ a positive integer.
Then $\thin(G) \geq k$ (resp. $\pthin(G) \geq k$) if and only if,
for every ordering ${<}$ of $V(G)$, the graph $G_{<}$ (resp.
$\tilde{G}_{<}$) has a clique of size $k$.
\end{corollary}

We will define also two new concepts related to (proper) thinness:
independent (proper) thinness and complete (proper) thinness.
These concepts are involved in some of the bounds of
Section~\ref{sec:thin-and-oper}.

A graph $G=(V,E)$ is \emph{$k$-independent-thin} if there exist an
ordering of $V$ and a partition of $V$ into $k$ classes,
consistent with the ordering, and such that each class is an
independent set of the graph.  The minimum $k$ such that $G$ is
$k$-independent-thin is called the \emph{independent thinness} of
$G$ and is denoted by $\indthin(G)$. Similarly, we can define the
concept of \emph{proper $k$-independent-thin} and
\emph{independent proper thinness} (denoted by $\indpthin(G)$),
where the partition has to be consistent with the ordering and its
reverse. Exchanging independent set by complete set, we define the
concepts of \emph{$k$-complete-thin}, \emph{complete thinness}
(denoted by $\compthin(G)$), \emph{proper $k$-complete-thin} and
\emph{complete proper thinness} (denoted by $\comppthin(G)$).

\begin{remark}\label{rem:co-comp-ind-thin} Notice that $\indpthin(G) \geq \indthin(G) \geq \chi(G)$
and, by Theorem~\ref{thm:thin-comp-order}, $\indpthin(G) =
\indthin(G) = \chi(G)$ when $G$ is a co-comparability graph.
Indeed, we can also see the independent (proper) thinness as a
coloring problem in a graph whose vertex set is $V(G)$ and whose
edge set is $E(G) \cup E(G_{<})$ (resp. $E(G) \cup
E(\tilde{G}_{<})$). Similarly, $\comppthin(G) \geq \compthin(G)
\geq \chi(\overline{G})$ and we can see the complete (proper)
thinness as a coloring problem in a graph whose vertex set is
$V(G)$ and whose edge set is $E(\overline{G}) \cup E(G_{<})$
(resp. $E(\overline{G}) \cup E(\tilde{G}_{<})$).
Theorem~\ref{thm:tK2} and Corollary~\ref{cor:crown} show that the
bounds $\compthin(G) \geq \chi(\overline{G})$ and $\indthin(G)
\geq \chi(G)$ can be arbitrarily bad. Notice also that, given a
(proper) $k$-thin representation of a graph, we can split each
class into independent sets and obtain a (proper)
$k$-independent-thin representation. Thus $\opindthin(G) \leq
\chi(G)\opthin(G)$. Analogously, $\opcompthin \leq
\chi(\overline{G})\opthin(G)$.
\end{remark}

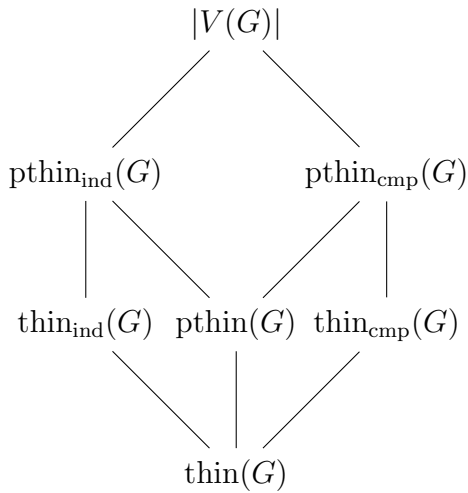
\begin{figure}[t]
\begin{center}
\begin{tikzpicture}
  \node (max) at (0,4) {$|V(G)|$};
  \node (a) at (-2,2) {$\indpthin(G)$};
  \node (c) at (2,2) {$\comppthin(G)$};
  \node (d) at (-2,0) {$\indthin(G)$};
  \node (e) at (0,0) {$\pthin(G)$};
  \node (f) at (2,0) {$\compthin(G)$};
  \node (min) at (0,-2) {$\thin(G)$};
  \draw (min) -- (d) -- (a) -- (max)
  (e) -- (min) -- (f) -- (c) -- (max);
  \draw[preaction={draw=white, -,line width=6pt}] (a) -- (e) -- (c);
\end{tikzpicture}
\end{center}
\caption{Hasse diagram of the parameters involved (if $\alpha$
precedes $\beta$ in the diagram, then $\alpha\leq \beta$). We will
state the strongest results, and the consequences for other
parameters can be deduced from the diagram.}\label{fig:hasse}
\end{figure}


\section{Thinness of some graph families and general
bounds}\label{sec:families}

In this section, we determine or give lower bounds for the
thinness and proper thinness of families of graphs, as induced
matchings, crowns, and hypercubes. In addition, we determine both
general lower and upper bounds for the thinness and proper
thinness of graphs. Also, we relate the thinness to the
independence and clique numbers of graphs.

For the complement of an induced matching, the exact value of the
thinness is known.

\begin{theorem}\cite{C-M-O-thinness-man}\label{thm:tK2} For every $t \geq
1$, $\thin(\overline{tK_2})=t$. \end{theorem}

The vertex partition used when proving the part
``$\thin(\overline{tK_2})=t$'' of the equation of
Theorem~\ref{thm:tK2}, which is consistent with any vertex
ordering, is the one where each class consists of a pair of
nonadjacent vertices. It is easy to see that this partition is
also strongly consistent with any vertex ordering. So we have the
following corollary.

\begin{corollary}\label{thm:tK2prop} For every $t \geq
1$,
$\pthin(\overline{tK_2})=\indthin(\overline{tK_2})=\indpthin(\overline{tK_2})=t$.
\end{corollary}




\subsection{Lower bounds}

{ Let $G_1$ and $G_2$ be graphs on $n$ vertices, and $f: V(G_1)
\to V(G_2)$ a bijection. Let $G_1 \boxminus_f G_2$ with vertex set
$V(G_1) \cup V(G_2)$, such that $V(G_i)$ induces $G_i$ for $i =
1,2$, and the edges from $V(G_1)$ to $V(G_2)$ are exactly
$\{vf(v)\}_{v \in V(G_1)}$. When $G_1$ or $G_2$ is either the
complete graph $K_n$ or the empty graph $nK_1$, we can omit $f$ by
symmetry. Notice that $\overline{tK_2} = \overline{tK_1 \boxminus
tK_1}$, and the \emph{crown} $CR_n = \overline{K_n \boxminus
K_n}$.

\begin{theorem}\label{thm:gcrown}
For $G_1$ and $G_2$ graphs on $n$ vertices and $f: V(G_1) \to
V(G_2)$ a bijection, $\thin(\overline{G_1 \boxminus_f G_2}) \geq
n/2$.
\end{theorem}

\begin{proof}
Let $G = \overline{G_1 \boxminus_f G_2}$, and $<$ an arbitrary
ordering of its vertices. We will show that $\omega(G_{<}) \geq
n/2$. Let $A = V(G_1)$ and $A'=V(G_2)$, and for each vertex $v \in
A$, let $v' = f(v)$ (the only vertex in $A'$ that is not adjacent
to $v$ in $G$).

By definition of $G_{<}$, if $v < v'$ then $v$ is adjacent in
$G_{<}$ to every vertex $w$ in $A$ such that $w < v$, and to every
vertex $w$ in $A$ such that $v < w < w'$. Analogously, if $v' <
v$, then $v'$ is adjacent in $G_{<}$ to every vertex $w'$ in $A'$
such that $w' < v'$, and to every vertex $w'$ in $A'$ such that
$v' < w' < w$. Therefore, the vertices $v \in A$ such that $v <
v'$ form a clique in $G_{<}$, and the vertices $v' \in A'$ such
that $v' < v$ form a clique in $G_{<}$. Since for each of the $n$
pairs of vertices $v,v'$ one of the inequalities holds, by the
pigeonhole principle, $G_{<}$ has a clique of size at least $n/2$.
\end{proof}

\begin{corollary}\label{cor:crown}
For every $n \geq 1$, $\thin(CR_n) \geq n/2$.
\end{corollary}

Since a \emph{fat spider} is the graph $\overline{K_n \boxminus
nK_1}$, the theorem above implies that split graphs have unbounded
thinness.}

\medskip

The \emph{vertex isoperimetric peak} of a graph $G$, denoted as
$b_v(G)$, is defined as $b_v(G) = \max_s \min_{X\subset V, |X|=s}
|N(X) \cap (V(G) \setminus X)|$, i.e., the maximum over $s$ of the
lower bounds for the number of boundary vertices (vertices outside
the set with a neighbor in the set) in sets of size $s$.

\begin{theorem}\cite{C-M-O-thinness-man}\label{thm:peak}
For every graph $G$ with at least one edge, $\thin(G)\geq
b_v(G)/\Delta(G)$.
\end{theorem}

The thinness of the grid $GR_r$ was lower bounded by using
Theorem~\ref{thm:peak}.

\begin{corollary}\cite{C-M-O-thinness-man}\label{cor:grid} For every $r \geq
2$, $\thin(GR_r) \geq r/4$.
\end{corollary}

We will prove next some other lower bounds, which are very useful
for bounding the thinness of highly symmetric graphs, as is the
case of graph products of highly symmetric graphs.


\begin{theorem}\label{thm:degree}
Let $G$ be a graph. If $|N(u) \setminus N[v]| \geq k$ for all $u,v
\in V(G)$ then $\thin(G) \geq k+1$. Moreover, for every order $<$
of $V(G)$, the first $k+1$ vertices induce a complete graph in
$G_{<}$.
\end{theorem}

\begin{proof}
Let $v_1, \ldots, v_n$ be an ordering of the vertices of $G$. Let
$i,j$ be such that $1 \leq i < j \leq k+1$. We know that $|N(v_i)
\setminus N[v_j]| \geq k$. Hence, $|(N(v_i) \setminus N[v_j])
\setminus (\{v_1, \ldots, v_{j-1}\} \setminus \{v_i\})| \geq 1$.
Therefore, there exists a vertex $v_h$ with $h > j$ such that $v_h
\in N(v_i)$ and $v_h \notin N(v_j)$, implying that $v_i$ and $v_j$
are adjacent in $G_{<}$.

So, for every order $<$ of vertices of $G$, we have that the first
$k+1$ vertices induce a complete graph in $G_{<}$. By
Corollary~\ref{cor:thin-comp-order}, $\thin(G) \geq k+1$.
\end{proof}



\begin{corollary}\label{cor:k-reg}
Let $G$ be a graph with $\delta(G) \geq d$ and such that for all
$u,v \in V(G)$, $|N(u) \cap N(v)| \leq c < d$. Then $\thin(G) \geq
d-c$.
\end{corollary}


The class of hypercubes $Q_n$ consists of the graphs whose vertex
sets correspond to all binary strings with fixed size $n$ and two
vertices $u$ and $v$ are adjacent if $u$ and $v$ differ exactly in
one position. We say that $n$ is the \emph{dimension} of the
hypercube $Q_n$. Clearly, $Q_n$ is a $n$-regular graph.

\begin{lemma}\cite{Mulder-Qn}
For all $u,v\in V(Q_n)$, $|N(u) \cap N(v)| \leq 2$.
\end{lemma}

\begin{corollary}\label{cor:Qn}
For every $n \geq 1$, $\thin(Q_n) \geq n-2$.
\end{corollary}

\begin{theorem}
Let $G$ be a graph. Let $S \subseteq V(G)$ and $p = |S|$. If
$|N(u) \setminus N[v]| \geq k$ for all $u,v \in S$ and $|V(G)| - p
\leq k$ then $\thin(G) \geq 1 + k + p - |V(G)|$.
\end{theorem}

\begin{proof}
Let $v_1, \ldots, v_n$ be an ordering of the vertices of $G$. Let
{$v_i,v_j \in S$ be} such that $1 \leq i < j \leq k+1$.

We know that $|N(v_i) \setminus N[v_j]| \geq k$. Hence, $|(N(v_i)
\setminus N[v_j]) \setminus (\{v_1, \ldots, v_{j-1}\} \setminus
\{v_i\})| \geq 1$. Therefore, there exists a vertex $v_h$ with $h
> j$ such that $v_h \in N(v_i)$ and $v_h \notin N(v_j)$, implying
that $v_i$ and $v_j$ are adjacent in $G_{<}$.

So, for every order $<$ of vertices of $G$, we have that the
vertices in $S$ within the first $k+1$ vertices induce a complete
graph in $G_{<}$. Since they are at least {$k+1-|V(G) \setminus S|
=  k+1-|V(G)|+p$, by Corollary~\ref{cor:thin-comp-order},
$\thin(G) \geq k+1-|V(G)|+p$.}
\end{proof}



\subsection{Upper Bounds}


{Two general upper bounds were known for the thinness of a graph.

\begin{theorem}\cite{C-M-O-thinness-man}\label{thm:n-log4}
Let $G$ be a graph. Then $\thin(G) \leq |V(G)| - \log(|V(G)|)/4$.
\end{theorem}

\begin{theorem}\cite{C-M-O-thinness-man}\label{thm:bound-delta}
Let $G$ be a graph. Then $\thin(G) \leq
|V(G)|(\Delta(G)+3)/(\Delta(G)+4)$. \end{theorem}

We will prove here other general upper bounds.}

\begin{lemma}\label{n-S+thin}
Let $S \subseteq V(G)$. Then $\thin(G) \leq |V(G)| - |S| +
\thin(G[S])$.
\end{lemma}

\begin{proof}
Consider an order $<$ of vertices of $G$ such that $v < s$ for all
$v \in V(G) - S$ and $s \in S$, and such that $\omega(G[S]_{<_S})
= \thin(G[S])$, where $<_S$ stands for the order restricted to
$S$. Such an order exists by Corollary~\ref{cor:thin-comp-order}.
Notice that, since $v < s$ for all $v \in V(G) - S$ and $s \in S$,
$G_{<}[S] = G[S]_{<_S}$. Then $\thin(G) \leq \omega(G_{<}) \leq
\omega(G_{<}[S]) + \omega(G_{<}[V(G) \setminus S]) \leq
\omega(G[S]_{<_S}) + |V(G)| - |S| = |V(G)| - |S| + \thin(G[S])$.
\end{proof}

\begin{corollary}\label{cor:subgint}
Let $S \subseteq V(G)$ be such that $G[S]$ is an interval graph.
Then $\thin(G) \leq |V(G)| - |S| + 1$.
\end{corollary}


An \emph{interval completion} of a graph $G$ is a spanning
supergraph of $G$ which is an interval graph.

\begin{lemma}
Let $G$ be a graph. Let $H$ be an interval completion of $G$. Let
$F$ be the subgraph of $H$ whose edges are $E(H) - E(G)$. Then
{the number of vertices of a maximum induced interval subgraph of
$G$}  is at least $|V(G)| - \tau(F)$.
\end{lemma}

\begin{proof}
Let $H$ be an interval completion of $G$. Then $H$ has $|V(G)|$
vertices. Let $F$ be the subgraph of $H$ whose edges are $E(H) -
E(G)$. If we remove from $H$ the vertices of a vertex cover in
$F$, we get an interval graph that is an induced subgraph of $G$.
\end{proof}

\begin{corollary}
Let $G$ be a graph. Let $H$ be an interval completion of $G$. Let
$F$ be the subgraph of $H$ whose edges are $E(H) - E(G)$. Then
$\thin(G) \leq \tau(F) + 1$.
\end{corollary}

In particular, stable and complete sets induce interval graphs.
Moreover, if $\alpha(G) < |V(G)|$ (resp. $\omega(G) < |V(G)|$), we
can add one more vertex $v$ and reorder the vertices of the stable
or complete set $S$ such that $u < w$ for all $u \in S - N(v)$ and
$w \in N(s) \cap S$ and $s < v$ for all $s \in S$, so $G$ has an
induced interval graph of size at least $\alpha(G) + 1$ (resp.
$\omega(G) +1$). As corollaries of Corollary~\ref{cor:subgint}, we
have the following two results.

\begin{corollary}\label{cor:thin-alpha}
If $V(G)$ is not a stable set then $\thin(G) \leq |V(G)| -
\alpha(G)$.
\end{corollary}

\begin{corollary}\label{cor:thin-omega}
If $G$ is not a complete graph then $\thin(G) \leq |V(G)| -
\omega(G)$.
\end{corollary}

\begin{remark}\label{rem:bounds-indep}
In the same way, one can see that  $\indthin(G) \leq |V(G)| -
\alpha(G)+1$ and $\compthin(G) \leq  |V(G)| - \omega(G)+1$ (in
this case we cannot add another vertex to the class containing the
maximum stable set or maximum clique, respectively).
\end{remark}

We have also the following bound for co-comparability graphs. As
already noticed in~\cite{B-M-O-thin-tcs},
Theorem~\ref{thm:thin-comp-order} implies that if $G$ is a
co-comparability graph, then $\thin(G) \leq \chi(G)$. Recalling
that co-comparability graphs are perfect~\cite{Meyn-co-comp}, this
implies $\thin(G) \leq \omega(G)$.

We prove next a new upper bound for the thinness of
co-comparability graphs.

\begin{theorem}
If $G$ is a non trivial co-comparability graph, then $\thin(G)
\leq |V(G)|/2$.
\end{theorem}

\begin{proof}
If $G$ is complete and non trivial, then $\thin(G) = 1 \leq
|V(G)|/2$. If $G$ is not complete, by
Corollary~\ref{cor:thin-omega}, $\thin(G) \leq |V(G)| -
\omega(G)$. Adding this inequality to the inequality $\thin(G)
\leq \omega(G)$ that holds for co-comparability graphs, we have
$2\thin(G) \leq |V(G)|$, thus $\thin(G) \leq |V(G)|/2$.
\end{proof}

The bound is attained, for example, by the family
$\overline{tK_2}$ (Theorem~\ref{thm:tK2}).


\section{Thinness and binary graph operations}\label{sec:thin-and-oper}

In this section, we analyze the behavior of the thinness and
proper thinness under different binary graph operations. {Each one
of these operations will be defined over a pair of graphs $G_1 =
(V_1,E_1)$ and $G_2 = (V_2,E_2)$ such that $|V_1| = n_1$, $|V_2| =
n_2$ and $V_1 \cap V_2 = \emptyset$. Besides, for some of the
following proofs, we consider an implicit ordering and partition
for both $V_1$ and $V_2$, as defined next.

The ordering of $V_1$ will be denoted by $v_1,\dots, v_{n_1}$ and
that of $V_2$ by $w_1,\dots, w_{n_2}$. Moreover, if the value
$t_i$ of some variation of thinness of $G_i$ (for $i \in \{1,2\}$)
is involved in the bound to be proved, the implicit ordering is
one consistent, according to the specified variation of thinness,
with a partition $(V_i^1, \ldots, V_1^{t_i})$. If, otherwise, only
the cardinality $n_i$ of $V_i$ is involved in the bound, the
implicit ordering is an arbitrary one. For instance, if $G_1$ is a
proper $t_1$-independent-thin graph, and $t_1$ is involved in the
bound to be proved, it means that the implicit ordering and
partition of $V_1$ are strongly consistent and all the $t_1$ parts
of the partition are independent sets.}

{Although the proofs in this section are not exactly the same,
some of them indeed share a common structure in the reasoning. For
the sake of conciseness, some of them were omitted and can be
found in~\ref{apdx:sec4:proofs}. }

\subsection{Union and join}\label{sec:unionjoin}

The \emph{union} of $G_1$ and $G_2$ is the graph $G_1 \cup G_2 =
(V_1 \cup V_2, E_1 \cup E_2)$, and the \emph{join} of $G_1$ and
$G_2$ is the graph {$G_1 \vee G_2 = (V_1 \cup V_2, E_1 \cup E_2
\cup \{vv' : v \in V_1, v' \in V_2\})$} (i.e., $\overline{G_1\vee
G_2} = \overline{G_1} \cup \overline{G_2}$). (The join is
sometimes also noted by $G_1 \otimes G_2$, but we follow the
notation in~\cite{B-D-thinness}).

The class of \emph{cographs} can be defined as the graphs that can
be obtained from trivial graphs by the union and join
operations~\cite{CorneilLerchsStewart81}. Aiming to characterize
$k$-thin graphs by forbidden induced subgraphs within the class of
cographs, the following results were proved.

\begin{theorem}\label{thm:union}\cite{B-D-thinness}
Let $G_1$ and $G_2$ be graphs. Then $f(G_1 \cup G_2) =
\max\{f(G_1),$ $f(G_2)\}$, for $f \in \{\thin,\pthin\}$.
\end{theorem}

\begin{theorem}\label{thm:join}\cite{B-D-thinness}
Let $G_1$ and $G_2$ be graphs. Then $f(G_1 \vee G_2) \leq
f(G_1)+f(G_2)$, for $f \in \{\thin,\pthin\}$. Moreover, if $G_2$
is complete, then $\thin(G_1 \vee G_2) = \thin(G_1)$.
\end{theorem}

\begin{lemma}\label{lem:join}\cite{B-D-thinness}
If $G$ is not complete, then $\thin(G \vee 2K_1) = \thin(G)+1$.
\end{lemma}

Lemma~\ref{lem:join} implies that if there is some constant value
$k$ such that recognizing $k$-thin graphs is NP-complete, then for
every $k' > k$, recognizing $k'$-thin graphs is NP-complete. The
existence of such $k$ is still not known, and in general the
complexity of recognition of $k$-thin and proper $k$-thin graphs,
both with $k$ as a parameter and with constant $k$, is open.

\begin{remark}\label{rem:clique-top} By definition of $G_{<}$, every non-smallest vertex of
any non-trivial clique has a vertex in $V(G)$ greater than it and
non-adjacent to it in $G$.
\end{remark}

The following lemma is necessary to prove Theorem~\ref{thm:join2}.

\begin{lemma}\label{lem:clique-top}
Let $G=(V,E)$ be a graph with $\thin(G)=k$ and $v_1 < \dots < v_n$
be an ordering of $V$. If $G$ is not complete, then there exist a
clique of size $k$ of $G_{<}$, $v_{i_1} < \dots < v_{i_k}$, and
$v_j
> v_{i_1}$, such that $v_j v_{i_1} \not \in E$.
\end{lemma}

\begin{proof}
By Corollary~\ref{cor:thin-comp-order}, for every order $<$ of the
vertices of $G$, $\omega(G_{<}) \geq k$. If $\thin(G)=1$, the
statement follows because $G$ is not complete. Suppose $\thin(G) >
1$, and let $v_1 < \dots < v_n$ be an ordering of $V$. By
definition of $G_{<}$, for every clique of $G_{<}$, all the
vertices that are not the smallest one have a vertex in $V(G)$
which is greater than it and non-adjacent to it in $G$. So, if
$G_{<}$ contains a clique of size greater than $k$, the statement
follows. {Consider} now that $\omega(G_{<})=k$. {In order to reach
a contradiction, suppose that} no clique of size $k$ of $G_{<}$
satisfies the property, then each vertex in the {set
$S = \{v \in V(G_<):v\text{ is the first vertex in $<$ of a clique of size $k$ of $G_<$}\}$} 
 is adjacent to every vertex greater than it in $G$. So, modifying the
order by placing $S$ as the largest vertices produces a graph
$G_{<'}$ which is a subgraph of $G_{<}$ and {in which the vertices
 of $S$ are isolated vertices}. In particular, $\omega(G_{<'}) <
\omega(G_{<}) = k$, a contradiction since, {by
Corollary~\ref{cor:thin-comp-order}, $\omega(G_{<'}) \geq k $}.
\end{proof}

We strength the result of Theorem~\ref{thm:join} for thinness.

\begin{theorem}\label{thm:join2}
Let $G_1$ and $G_2$ be graphs. If $G_1$ is complete, then
$\thin(G_1 \vee G_2) = \thin(G_2)$. If neither $G_1$ nor $G_2$ are
complete, then $\thin(G_1 \vee G_2) = \thin(G_1) + \thin(G_2)$.
\end{theorem}

\begin{proof}
Let $G = G_1 \vee G_2$. If one of them is complete (suppose
without loss of generality $G_1$), then, by Theorem
\ref{thm:join}, $\thin(G) = \thin(G_2)$. Otherwise, by Theorem
\ref{thm:join}, $\thin(G_1 \vee G_2) \leq \thin(G_1)+\thin(G_2)$.
Let us prove the equality. Let $k_1 = \thin(G_1)$, $k_2 = \thin
(G_2)$, and $k = \thin(G_1 \vee G_2)$. Let $<$ be an ordering
consistent with a $k$-partition of $V(G)$. Let ${G_1}_{<}$ and
${G_2}_{<}$ {be the incompatibility graphs} obtained from the
order $<$ restricted to $V_1$ and $V_2$, respectively.

By Lemma \ref{lem:clique-top}, there exist a $k_1$-clique of
${G_1}_{<}$, $v^1_{i_1} < \dots < v^1_{i_{k_1}}$, and $v^1_j
> v^1_{i_1}$, such that $v^1_j v^1_{i_1} \not \in E_1$. As
well, there exist a $k_2$-clique of ${G_2}_{<}$, $v^2_{i_1} <
\dots < v^2_{i_{k_2}}$, and $v^2_j
> v^2_{i_1}$, such that $v^2_j v^2_{i_1} \not \in E_2$.

Notice also that, by definition of ${G_i}_{<}$, $i=1,2$, every
non-smallest vertex of the clique of ${G_i}_{<}$ has a vertex in
$V_i$, greater than it and non-adjacent to it in $G_i$.
Considering this property and the fact that every vertex of $G_1$
is adjacent to every vertex of $G_2$ in $G$, it follows that every
vertex of $\{v^1_{i_1}, \dots, v^1_{i_{k_1}}\}$ is adjacent to
every vertex of $\{v^2_{i_1}, \dots, v^2_{i_{k_2}}\}$ in $G_{<}$,
hence $k \geq k_1+k_2$. This completes the proof of the theorem.
\end{proof}

\begin{remark}
The proper thinness of the join $G_1 \vee G_2$ cannot be expressed
as a function whose only parameters are the proper thinness of
$G_1$ and $G_2$ (even excluding simple particular cases, like
trivial or complete graphs). The graph $tK_1$ has proper
thinness~1 for every $t$. By Theorem~\ref{thm:join}, the proper
thinness of the join of two graphs of proper thinness~1 is either
1~or~2, and there are examples for both of the values. The graph
$P_3 = 2K_1 \vee K_1$ has proper thinness~1 but $3K_1 \vee K_1$,
known as \emph{claw}, or $C_4 = 2K_1 \vee 2K_1$ have proper
thinness~2 (the claw and $C_4$ are not proper interval graphs).
Similarly, by Theorem~\ref{thm:join}, the proper thinness of the
join of a graph of proper thinness~2 and a graph of proper
thinness~1 is either 2~or~3, and there are examples for both of
the values. The graph $(\mbox{claw} \cup tK_1) \vee K_1$ has
proper thinness~2, but the graph $3\mbox{claw} \vee K_1$ has
proper thinness~3~\cite{B-D-thinness}.
\end{remark}

Nevertheless, we have a lemma similar to Lemma~\ref{lem:join} for
proper thinness.

\begin{lemma}\label{lem:joinp}\cite{B-D-thinness}
For every graph $G$, $\pthin(3G \vee K_1) = \pthin(G)+1$.
\end{lemma}

\begin{proof}
Theorems~\ref{thm:union} and~\ref{thm:join} imply $\pthin(3G \vee
K_1) \leq \pthin(G)+1$. To show $\pthin(3G \vee K_1) \geq
\pthin(G)+1$, we will use Corollary~\ref{cor:thin-comp-order}. Let
$u$ be the corresponding vertex of the $K_1$ in $H=3G \vee K_1$,
and let $<$ be an ordering of the vertices of $H$. Let $w,w'$ be
the minimum and maximum vertices, respectively, according to $<$
restricted to $H \setminus \{u\}$. Let $S$ be the set of vertices
of the copy of $G$ in $H$ which contains neither $w$ nor $w'$. We
will show that in $\tilde{H}_{<}$, all the vertices of $S$ are
adjacent to $u$. Let $v$ be a vertex of $S$. If $u < v$, then $u <
v < w'$, $w'v \not \in V(H)$ and $w'u \in V(H)$. If $v < u$, then
$w < v < u$, $wv \not \in V(H)$ and $wu \in V(H)$. In either case,
by definition of $\tilde{H}_{<}$, $uv \in V(\tilde{H}_{<})$. By
Corollary~\ref{cor:thin-comp-order}, $\tilde{H}_{<}[S]$ contains a
clique of size $\pthin(G)$, and thus $\tilde{H}_{<}[S \cup \{u\}]$
contains a clique of size $\pthin(G)+1$. As the order $<$ was
arbitrary, again by Corollary~\ref{cor:thin-comp-order},
$\pthin(H) \geq \pthin(G)+1$.
\end{proof}

We will now study the behavior of independent and complete
(proper) thinness under the union and join operations.

\begin{theorem}\label{thm:union-ind}
Let $G_1$ and $G_2$ be graphs. Then $f(G_1 \cup G_2) =
\max\{f(G_1),$ $f(G_2)\}$, for $f \in \{\indthin, \indpthin\}$.
\end{theorem}

\begin{proof}
Since both  $G_1$ and $G_2$ are induced subgraphs of $G_1 \cup
G_2$, then $\indthin(G_1 \cup G_2) \geq
\max\{\indthin(G_1),\indthin(G_2)\}$ and the same holds for the
independent proper thinness.

Let $G_1$ and $G_2$ be two graphs with independent thinness (resp.
independent proper thinness) $t_1$ and $t_2$, respectively.
Suppose without loss of generality that $t_1 \leq t_2$. For $G =
G_1 \cup G_2$, define a partition $(V^1,\dots, V^{t_2})$ such that
$V^i = V_1^i \cup V_2^i$ for $i = 1, \dots, t_1$ and $V^i = V_2^i$
for $i = t_1+1, \dots, t_2$, and define $v_1,\dots,
v_{n_1},w_1,\dots, w_{n_2}$ as an ordering of the vertices. By
definition of union of graphs, the sets of the partition are
independent and, if three ordered vertices according to the order
defined in $V(G_1 \cup G_2)$ are such that the first and the third
are adjacent, either the three vertices belong to $V_1$ or the
three vertices belong to $V_2$. Since the order and the partition
restricted to each of $G_1$ and $G_2$ are the original ones, the
properties required for consistency (resp. strong consistency) are
satisfied.
\end{proof}

\begin{theorem}\label{thm:union-comp}
Let $G_1$ and $G_2$ be graphs. Then $f(G_1 \cup G_2) =
f(G_1)+f(G_2)$, for $f  \in \{\compthin,\comppthin\}$.
\end{theorem}

\begin{proof}
Let $G_1$ and $G_2$ be two graphs with complete thinness (resp.
complete proper thinness) $t_1$ and $t_2$, respectively. By
definition of $G_1 \cup G_2$, any vertex ordering and partition
into complete sets of $G_1 \cup G_2$ which are (strongly)
consistent are (strongly) consistent when restricted to $V_1$ and
$V_2$. Notice that no complete set of $G_1 \cup G_2$ contains both
a vertex of $V_1$ and a vertex of $V_2$. So, $\compthin(G_1 \cup
G_2)$ (resp. $\comppthin(G_1 \cup G_2)$) is at least $t_1 + t_2$.
On the other hand, consider orderings and partitions of $V_1$ and
$V_2$ into $t_1$ and $t_2$ complete sets, respectively, which are
consistent (resp. strongly consistent). For $G_1 \cup G_2$, define
a partition with $t_1+t_2$ complete sets as the union of the two
partitions and define as ordering of the vertices the
concatenation of the orderings of $V_1$ and $V_2$ ({i.e., $v < v'$
if either $v$ and $v'$ belong to $V_i$ and $v < v'$, for $i\in
\{1,2\}$, or $v \in V_1$ and $v'$ in $V_2$}). By definition of
union of graphs, if three ordered vertices according to the order
defined in $V(G_1 \cup G_2)$ are such that the first and the third
are adjacent, either the three vertices belong to $V_1$ or the
three vertices belong to $V_2$. Since the order and the partition
restricted to each of $G_1$ and $G_2$ are the original ones, the
properties required for consistency (resp. strong consistency) are
satisfied. Thus $\compthin(G_1 \cup G_2)$ (resp. $\comppthin(G_1
\cup G_2)$) is at most $t_1 + t_2$, which completes the proof.
\end{proof}

\begin{theorem}\label{thm:join-ind}
Let $G_1$ and $G_2$ be graphs. Then $f(G_1 \vee G_2) =
f(G_1)+f(G_2)$, for $f \in \{\indthin, \indpthin\}$.
\end{theorem}

\begin{proof}
Let $G_1$ and $G_2$ be two graphs with independent thinness (resp.
independent proper thinness) $t_1$ and $t_2$, respectively. By
definition of $G_1 \vee G_2$, any vertex ordering and partition
into independent sets of $G_1 \vee G_2$ which are (strongly)
consistent are (strongly) consistent when restricted to $V_1$ and
$V_2$. Notice that no independent set of $G_1 \vee G_2$ contains
both a vertex of $V_1$ and a vertex of $V_2$. So, $\indthin(G_1
\vee G_2)$ (resp. $\indpthin(G_1 \vee G_2)$) is at least $t_1 +
t_2$.

On the other hand, consider orderings and partitions of $V_1$ and
$V_2$ into $t_1$ and $t_2$ independent sets, respectively, which
are consistent (resp. strongly consistent). For $G = G_1 \vee
G_2$, define a partition with $t_1+t_2$ {independent sets} as the
union of the two partitions and define as ordering of the vertices
the concatenation of the orderings of $V_1$ and $V_2$. Let $x,y,z$
be three vertices of $V(G)$ such that $x < y < z$, $xz \in E(G)$,
and $x$ and $y$ are in the same class of the partition of $V(G)$.
Then, in particular, $x$ and $y$ both belong either to $V_1$ or to
$V_2$. If $z$ belongs to the same graph, then $yz \in E(G)$
because the ordering and partition restricted to each of $G_1$ and
$G_2$ are consistent. Otherwise, $z$ is also adjacent to $y$ by
the definition of join. We have proved that the defined partition
and ordering are consistent. The respective proof of the strong
consistency, given the strong consistency of the partition and
ordering of each of $G_1$ and $G_2$, is symmetric. Then
$\indthin(G_1 \vee G_2)$ (resp. $\indpthin(G_1 \vee G_2))$ is at
most $t_1+t_2$, which completes the proof.
\end{proof}

Lemma~\ref{lem:joinp} and Theorems~\ref{thm:union-comp}
and~\ref{thm:join-ind} imply the following corollary.

\begin{corollary}
If there is some constant value $k$ such that recognizing proper
$k$-thin graphs (resp. (proper) $k$-independent-thin and (proper)
$k$-complete-thin graphs) is NP-complete, then for every $k' > k$,
recognizing proper $k'$-thin graphs (resp. (proper)
$k'$-independent-thin and (proper) $k'$-complete-thin graphs) is
NP-complete.
\end{corollary}

Recall that the existence of such $k$ is still not known for all
these classes.

\begin{theorem}\label{thm:join-comp}
Let $G_1$ and $G_2$ be graphs. Then $f(G_1 \vee G_2) \leq
f(G_1)+f(G_2)$, for $f \in \{\compthin, \comppthin\}$. Moreover,
if $G_2$ is complete, then $\compthin(G_1 \vee G_2) =
\compthin(G_1)$.
\end{theorem}

\subsection{Graph composition or lexicographical product}\label{sec:lex}

{Let $v \in V_1$.} The \emph{lexicographical product} of $G_1$ and
$G_2$ \emph{over the vertex} $v$ is the graph $G_1 \bullet_v G_2$
obtained from $G_1$ by replacing vertex $v$ by graph $G_2$, i.e.,
$V(G_1 \bullet_v G_2) = V_2 \cup V_1 \setminus \{v\}$, and $x$,
$y$ are adjacent if either $x, y \in V_1 \setminus \{v\}$ and $xy
\in E_1$, or $x, y \in V_2$ and $xy \in E_2$, or $x \in V_1
\setminus \{v\}$, $y \in V_2$, and $xv \in E_1$.

\begin{theorem}\label{thm:lexv} Let $G_1$ and $G_2$ be two graphs.
Then $f(G_1 \bullet_v G_2) \leq f(G_1) + f(G_2)$, for $f \in
\{\thin,$ $\pthin,$ $\compthin,$ $\comppthin\}$, and $f(G_1
\bullet_v G_2) \leq f'(G_1) + f(G_2)-1$, for $(f,f') \in
\{(\thin,\indthin),$ $(\indthin,\indthin),$ $(\pthin,\indpthin),$
$(\indpthin,\indpthin)\}$. Moreover, if $G_2$ is complete, $f(G_1
\bullet_v G_2) = f(G_1)$, for $f \in \{\thin,$ $\pthin,$
$\compthin,$ $\comppthin\}$.
\end{theorem}

\begin{proof}
Let $G_1$ and $G_2$ be two graphs with thinness (resp. proper
thinness) $t_1$ and $t_2$, respectively.

For $G = G_1 \bullet_v G_2$, if $v$ is the $i$-th vertex in the
ordering of $V_1$ and belongs to the class $V_1^j$, define
$v_1,\dots,v_{i-1}, w_1,\dots, w_{n_2},v_{i+1}, \dots, v_{n_1}$ as
an ordering of the vertices of $G$, and a partition with at most
$t_1+t_2$ sets as the union of the two partitions, where $V_1^j$
is replaced by $V_1^j \setminus \{v\}$ (or eliminated if $v$ is
the only vertex in the class, justifying the partition to have
\emph{at most} $t_1+t_2$ classes).

Let $x,y,z$ be three vertices of $V(G)$ such that $x < y < z$, $xz
\in E(G)$, and $x$ and $y$ are in the same class of the partition
of $V(G)$. Then, in particular, $x$ and $y$ both belong either to
$V_1 \setminus \{v\}$ or to $V_2$. If $z$ belongs to the same
graph, then $yz \in E(G)$ because the ordering and partition
restricted to each of $G_1$ and $G_2$ are consistent.

Otherwise, if $x$ and $y$ belong to $V_2$ and $z$ belongs to $V_1
\setminus \{v\}$, then $z$ is adjacent to $y$ because $V_2$ is an
homogeneous set in $G$. If $x$ and $y$ belong to $V_1 \setminus
\{v\}$ and $z$ belongs to $V_2$, by the definition of the order in
$G$, $y < v$ in the order of $V_1$, so $v$ is adjacent to $y$ in
$G_1$. By the definition of $G$, $y$ is adjacent to $z$.

We have proved that the defined partition and ordering are
consistent, and thus that $\thin(G_1 \bullet_v G_2) \leq
\thin(G_1)+\thin(G_2)$. The proof of the strong consistency, given
the strong consistency of the partition and ordering of each of
$G_1$ and $G_2$, is symmetric and implies $\pthin(G_1 \bullet_v
G_2) \leq \pthin(G_1)+\pthin(G_2)$.

Notice that if the partitions of $G_1$ and $G_2$ are into complete
sets (resp. independent sets), so is the defined partition of $G_1
\bullet_v G_2$. Moreover, we will define a new partition with
$t_1+t_2-1$ sets as the union of the partitions of $V_1$ and
$V_2$, where $V_1^j$ and $V_2^1$ are replaced by $V_1^j \setminus
\{v\} \cup V_2^1$. Notice that if the partitions of $G_1$ and
$G_2$ are into complete sets (resp. independent sets), so is the
new class.

We will prove first that if the partition of $V_1$ consists of
independent sets (not necessarily the partition of $V_2$), then
the order and the new partition are consistent. For strongly
consistence the proof is symmetric.

The only cases that need to be re-considered are the ones in which
$x$ and $y$ belong to the new class and, moreover, one of them
belongs to $V_1$ and the other one belongs to $V_2$. If $x$
belongs to $V_2$ and $y$ belongs to $V_1$, since the vertices of
$G_2$ are consecutive in the order, this implies that $z$ belongs
to $V_1$. Since $xz \in E(G)$, $vz \in E_1$. By the order
definition, $v < y < z$ in $G_1$. Since the ordering and partition
of $G_1$ are consistent, $yz \in E_1$ and thus $yz \in E(G)$. If
$x$ belongs to $V_1$ and $y$ belongs to $V_2$, since the partition
of $G_1$ consists of independent sets, $xv \not \in V_1$. So, if
$xz \in E(G)$, $z$ belongs to $V_1$. By the order definition, $x <
v < z$ in $G_1$. Since the ordering and partition of $G_1$ are
consistent, $vz \in E_1$ and thus $yz \in E(G)$.

Now we will prove that if $G_2$ is complete, i.e., $V_2 = V_2^1$,
then the order and the new partition are consistent. For strongly
consistence, the proof is symmetric.

Again, the only cases that need to be re-considered are the ones
in which $x$ and $y$ belong to the new class and, moreover, one of
them belongs to $V_1$ and the other one belongs to $V_2$. We will
prove them for consistence, and for strongly consistence the proof
is symmetric.

If $x$ belongs to $V_2$ and $y$ belongs to $V_1$, since the
vertices of $G_2$ are consecutive in the order, this implies that
$z$ belongs to $V_1$. Since $xz \in E(G)$, $vz \in E_1$. By the
order definition, $v < y < z$ in $G_1$. Since the ordering and
partition of $G_1$ are consistent, $yz \in E_1$ and thus $yz \in
E(G)$. If $x$ belongs to $V_1$ and $y$ belongs to $V_2$, since
$G_2$ is complete, if $z \in V_2$, $zy \in V(G)$. So, assume $z$
belongs to $V_1$. By the order definition, $x < v < z$ in $G_1$.
Since the ordering and partition of $G_1$ are consistent, $xz \in
E_1$ implies $vz \in E_1$ and thus $yz \in E(G)$.
\end{proof}

\begin{remark} Notice that if $v$ is isolated in $G_1$, then $G_1 \bullet_v G_2 = G_1[V(G_1) \setminus \{v\}] \cup G_2$, and if $v$ is universal in $G_1$, then $G_1 \bullet_v G_2 = G_1[V(G_1) \setminus \{v\}] \vee G_2$, so we can obtain better bounds by using the results of Section~\ref{sec:unionjoin}. \end{remark}

An equivalent formulation of Theorem~\ref{thm:lexv} is the
following.

\begin{theorem}\label{thm:lexhom}
Let $H$ be an homogeneous set of $G$, and $G|_H$ be the graph
obtained by contracting $H$ into a vertex. Then $f(G) \leq f(G|_H)
+ f(H)$, for $f \in \{\thin,$ $\pthin,$ $\compthin,$
$\comppthin\}$, and $f(G) \leq f'(G|_H) + f(H)-1$, for $(f,f') \in
\{(\thin,\indthin),$ $(\indthin,\indthin),$ $(\pthin,\indpthin),$
$(\indpthin,\indpthin)\}$. Moreover, if $H$ is complete, $f(G) =
f(G|_H)$, for $f \in \{\thin,$ $\pthin,$ $\compthin,$
$\comppthin\}$.
\end{theorem}

The \emph{lexicographical product} of $G_1$ and $G_2$ (also known
as \emph{composition} of $G_1$ and $G_2$) is the graph $G_1
\bullet G_2$ (also noted as $G_1[G_2]$) whose vertex set is the
Cartesian product $V_1 \times V_2$, and such that two vertices
$(u_1,u_2)$ and $(v_1,v_2)$ are adjacent in $G_1 \bullet G_2$ if
and only if either $u_1 = v_1$ and $u_2$ is adjacent to $v_2$ in
$G_2$, or $u_1$ is adjacent to $v_1$ in $G_1$. It is not
necessarily commutative.

\begin{theorem}\label{thm:lex} Let $G_1$ and $G_2$ be two graphs.
Then, if $G_2$ is complete, $f(G_1 \bullet G_2) = f(G_1)$, for $f
\in \{\thin,$ $\pthin,$ $\compthin,$ $\comppthin\}$. Also, $f(G_1
\bullet G_2) \leq f'(G_1)f(G_2)$, for $(f,f') \in
\{(\thin,\indthin),$ $(\indthin,\indthin),$ $(\pthin,\indpthin),$
$(\indpthin,\indpthin)\}$, and $f(G_1 \bullet G_2) \leq
|V_1|f(G_2)$, for $f \in \{\compthin,$ $\comppthin\}$. If $G_2$ is
not complete, $\omega(G_1)f(G_2) \leq f(G_1 \bullet G_2)$, for $f
\in \{\thin,\indthin,\allowbreak \indpthin\}$.
\end{theorem}

\begin{proof}
If $G_2$ is complete, we can iteratively apply
Theorem~\ref{thm:lexv}, since $G_1 \bullet G_2 = ((\dots((G_1
\bullet_{v_1} G_2) \bullet_{v_2} G_2) \dots ) \bullet_{v_{n_1}}
G_2)$, with {$\{v_1,\dots, v_{n_1}\} = V_1$}. By induction in
$n_1$, $f(G_1 \bullet G_2) = f(G_1)$, for $f \in \{\thin,$
$\pthin,$ $\compthin,$ $\comppthin\}$.


So, let $G_1$ and $G_2$ be two graphs, such that $f'(G_1) = t_1$
and $f(G_2) = t_2$, and assume $G_2$ is not complete. In $G_1
\bullet G_2$, consider $V_1 \times V_2$ lexicographically ordered
with respect to the {defined} orderings of $V_1$ and $V_2$.

Consider first $(f,f') \in \{(\thin,\indthin),$
$(\indthin,\indthin),$ $(\pthin,\indpthin),$
$(\indpthin,\indpthin)\}$, and the partition $\{V^{i,j}\}_{1\leq i
\leq t_1,\ 1 \leq j \leq t_2}$ such that $V^{i,j} = \{(v,w) : v
\in V_1^i, w \in V_2^j\}$ for each $1\leq i \leq t_1$, $1 \leq j
\leq t_2$. Since the partition of $V_1$ consists of independent
sets, vertices $(v,w)$ and $(v',w)$ in the same partition are not
adjacent for $v \neq v'$, and if furthermore the partition of
$V_2$ consists of independent sets, the same property holds for
the defined partition of $V_1 \times V_2$ for $G_1 \bullet G_2$.

We will show now that this ordering and partition of $V_1 \times
V_2$ are consistent (resp. strongly consistent, when $f, f'$ are
proper). Let $(v_p,w_i), (v_q,w_j), (v_r,w_{\ell})$ be three
vertices appearing in that ordering in $V_1 \times V_2$.

\emph{Case 1: $p = q = r$.} In this case, $i < j < \ell$. Suppose
first that $(v_p,w_i), (v_q,w_j)$ belong to the same class, i.e.,
$w_i, w_j$ belong to the same class in $G_2$. Vertices $(v_p,w_i)$
and $(v_r,w_{\ell})$ are adjacent in $G_1 \bullet G_2$ if and only
if $w_i w_{\ell} \in E_2$. Since the order and partition of $G_2$
are consistent, $w_j w_{\ell} \in E_2$, so $(v_q,w_j)$ and
$(v_r,w_{\ell})$ are adjacent in $G_1 \bullet G_2$, as required.
If $f$ and $f'$ are proper, the proof for strongly consistence is
symmetric.

\emph{Case 2: $p = q < r$.} Suppose first that $(v_p,w_i),
(v_q,w_j)$ belong to the same class, i.e., $w_i, w_j$ belong to
the same class in $G_2$. Vertices $(v_p,w_i)$ and $(v_r,w_{\ell})$
are adjacent in $G_1 \bullet G_2$ if and only if $v_p v_r \in
E_1$. Since $p = q$, $v_q v_r \in E_1$, so $(v_q,w_j)$ and
$(v_r,w_{\ell})$ are adjacent in $G_1 \bullet G_2$, as required.
Suppose now $f$ and $f'$ are proper, and {$(v_q, w_j), (v_r,
w_{\ell})$} belong to the same class, and, in particular, $v_q,
v_r$ belong to the same class in $G_1$ thus they are not adjacent.
Since $p = q$, $(v_p,w_i)$ and $(v_r, w_{\ell})$ are not adjacent
in $G_1 \bullet G_2$.

\emph{Case 3: $p < q = r$.} Suppose first that $(v_p,w_i),
(v_q,w_j)$ belong to the same class, and, in particular, $v_p,
v_q$ belong to the same class in $G_1$. Thus, they are not
adjacent. Since $q = r$, $(v_p,w_i)$ and $(v_r, w_{\ell})$ are not
adjacent in $G_1 \bullet G_2$. Suppose now $f$ and $f'$ are
proper, and $(v_q, w_j), (v_r, w_{\ell})$ belong to the same
class, i.e., $w_i, w_j$ belong to the same class in $G_2$.
Vertices $(v_p,w_i)$ and $(v_r,w_{\ell})$ are adjacent in $G_1
\bullet G_2$ if and only if $v_p v_r \in E_1$. Since $q = r$, $v_p
v_q \in E_1$, so $(v_p,w_i)$ and $(v_q,w_j)$ are adjacent in $G_1
\bullet G_2$, as required.

\emph{Case 4: $p < q < r$.} Suppose first that $(v_p,w_i),
(v_q,w_j)$ belong to the same class, and, in particular, $v_p,
v_q$ belong to the same class in $G_1$. Since the ordering an the
partition of $G_1$ are consistent, if $(v_p,w_i), (v_r, w_{\ell})$
are adjacent in $G_1 \bullet G_2$, in particular $v_p, v_r$ are
adjacent in $G_1$, thus $v_q, v_r$ are adjacent in $G_1$ and
$(v_q,w_j), (v_r, w_{\ell})$, as required. If $f$ and $f'$ are
proper, the proof for strongly consistence is symmetric.

Consider $f \in \{\compthin,$ $\comppthin\}$, and the partition
$\{V^{i,j}\}_{1\leq i \leq n_1,\ 1 \leq j \leq t_2}$ such that
$V^{i,j} = \{(v_i,w) : w \in V_2^j\}$ for each $1\leq i \leq n_1$,
$1 \leq j \leq t_2$. Since the partition of $V_2$ consists of
complete sets, the same property holds for the defined partition
of $V_1 \times V_2$ for $G_1 \bullet G_2$.

We will show now that this ordering and partition of $V_1 \times
V_2$ are consistent (resp. strongly consistent). Let $(v_p,w_i),
(v_q,w_j), (v_r,w_{\ell})$ be three vertices appearing in that
ordering in $V_1 \times V_2$.

\emph{Case 1: $p = q = r$.} In this case, $i < j < \ell$. Suppose
first that $(v_p,w_i), (v_q,w_j)$ belong to the same class, i.e.,
$w_i, w_j$ belong to the same class in $G_2$. Vertices $(v_p,w_i)$
and $(v_r,w_{\ell})$ are adjacent in $G_1 \bullet G_2$ if and only
if $w_i w_{\ell} \in E_2$. Since the order and partition of $G_2$
are consistent, $w_j w_{\ell} \in E_2$, so $(v_q,w_j)$ and
$(v_r,w_{\ell})$ are adjacent in $G_1 \bullet G_2$, as required.
If $f = \comppthin$, the proof for strongly consistence is
symmetric.

\emph{Case 2: $p = q < r$.} Suppose first that $(v_p,w_i),
(v_q,w_j)$ belong to the same class, i.e., $w_i, w_j$ belong to
the same class in $G_2$. Vertices $(v_p,w_i)$ and $(v_r,w_{\ell})$
are adjacent in $G_1 \bullet G_2$ if and only if $v_p v_r \in
E_1$. Since $p = q$, $v_q v_r \in E_1$, so $(v_q,w_j)$ and
$(v_r,w_{\ell})$ are adjacent in $G_1 \bullet G_2$, as required.
No further restriction has to be satisfied if $f = \comppthin$,
since by definition of the classes $(v_q, w_j), (v_r, w_{\ell})$
belong to different classes.

\emph{Case 3: $p < q = r$.} No restriction has to be satisfied for
consistence, as $(v_p,w_i), (v_q,w_j)$ belong to different
classes. If $f = \comppthin$, suppose that $(v_q, w_j), (v_r,
w_{\ell})$ belong to the same class, i.e., $w_j, w_{\ell}$ belong
to the same class in $G_2$. Vertices $(v_p,w_i)$ and
$(v_r,w_{\ell})$ are adjacent in $G_1 \bullet G_2$ if and only if
$v_p v_r \in E_1$. Since $q = r$, $v_p v_q \in E_1$, so
$(v_p,w_i)$ and $(v_q,w_j)$ are adjacent in $G_1 \bullet G_2$, as
required.

\emph{Case 4: $p < q < r$.} In this case, the three vertices are
in different classes, so no restriction has to be satisfied.


To prove the lower bound when $G_2$ is not complete, notice that
$K_r \bullet G_2$ is isomorphic to $(((G_2 \vee G_2) \vee G_2)
\dots \vee G_2)$ ($r$ times). By Theorems~\ref{thm:join2}
and~\ref{thm:join-ind}, $\omega(G_1)f(G_2) \leq f(G_1 \bullet
G_2)$, for $f \in \{\thin,\indthin,\indpthin\}$.
\end{proof}

\begin{corollary}\label{thm:lexomega}
Let $G_1$ and $G_2$ be graphs. If $G_2$ is not complete, then
$\thin(G_1 \bullet G_2) \geq \omega(G_1)$.
\end{corollary}

Notice that $K_n \bullet 2K_1 = \overline{tK_2}$. So, we have the
following corollary of Theorem~\ref{thm:tK2}.

\begin{corollary}\label{thm:nblex}
There is no function $f: \mathbb{R}^2 \to \mathbb{R}$ such that
$\thin(G_1 \bullet G_2) \leq f(\comppthin(G_1),|V(G_2)|)$ for any
pair of graphs $G_1$, $G_2$.
\end{corollary}

The non existence of bounds in terms of other parameters can be
deduced from diagram in Figure~\ref{fig:hasse}.

\begin{corollary}
Let $G_1$ be a co-comparability graph. If $G_2$ is complete, then
$\thin(G_1 \bullet G_2) = \thin(G_1)$, and if not, then $\thin(G_1
\bullet G_2) = \omega(G_1)\thin(G_2)$.
\end{corollary}

\begin{theorem}\label{thm:pthin-indpthin}
Let $G$ be a graph and $t \geq 3$, $q \geq 1$. Then, $\pthin(G
\bullet tK_1) = \indpthin(G)$, and $\pthin((G \bullet tK_1) \vee
qK_1) = \pthin((G \bullet tK_1) \vee K_q) = \indpthin(G)+1$.
\end{theorem}

\begin{proof}
The upper bounds are a consequence of Theorems~\ref{thm:lex}
and~\ref{thm:join}.

For the lower bounds, we will prove the statement for $t = 3$ and
$q = 1$, since $(G \bullet 3K_1)$ (resp. $(G \bullet 3K_1) \vee
K_1$) is an induced subgraph of $(G \bullet tK_1)$ (resp. $(G
\bullet tK_1) \vee qK_1$ and $(G \bullet tK_1) \vee K_q$). Let $G'
= (G \bullet 3K_1)$ and $G'' = G' \vee K_1$. Let $V(G') = \{v_i^1
< v_i^2 < v_i^3 : v_i \in V(G)\}$, and $V(G'') = V(G') \cup
\{u\}$. Consider an ordering of the vertices of $G''$, and let $<$
be the vertex order of $V(G)$ induced by the order of
$\{v_i^2\}_{v_i \in V(G)}$. We will show the following three
statements, that are enough to prove the theorem: if $vw \in
E(\tilde{G}_{<})$ then $v^2w^2 \in E(\tilde{G'}_{<})$; if $vw \in
E(G)$ then $v^2w^2 \in E(\tilde{G'}_{<})$; for any $v \in V(G)$,
$v^2u \in E(\tilde{G''}_{<})$.

First, let $v < w$ be adjacent in $\tilde{G}_{<}$. Then either
there is a vertex $z$ such that $v < w < z$, $vz \in E(G)$ and $wz
\not \in E(G)$, or there is a vertex $z$ such that $z < v < w$,
$zw \in E(G)$ and $zv \not \in E(G)$. In either case, the same
holds for $v^2, w^2, z^2$, so $v^2w^2 \in E(\tilde{G'}_{<})$.

Next, let $v < w$ be adjacent in $G$. Then $v^2 < w^2 < w^3$,
$v^2w^3 \in E(G')$ and $w^2w^3 \not \in E(G')$, so $v^2w^2 \in
E(\tilde{G'}_{<})$.

Finally, for $v \in V(G)$, if $u < v^2$, then $u < v^2 < v^3$,
$uv^3 \in E(G')$ and $v^2v^3 \not \in E(G')$, so $uv^2 \in
E(\tilde{G'}_{<})$. If $v^2 < u$, then $v^1 < v^2 < u$, $v^1u \in
E(G')$ and $v^1v^2 \not \in E(G')$, so $uv^2 \in
E(\tilde{G'}_{<})$.

By Remark~\ref{rem:co-comp-ind-thin}, it holds that $\pthin(G')
\geq \indpthin(G)$, and that $\pthin(G'') \geq \indpthin(G)+1$.
\end{proof}

Theorem~\ref{thm:pthin-indpthin} implies that if recognizing
proper $k$-independent-thin graphs is NP-complete, then
recognizing proper $k$-thin graphs is NP-complete (both with $k$
as a parameter and with constant $k$).

\subsection{Cartesian product}

The \emph{Cartesian product} $G_1 \ \square \ G_2$ is a graph
whose vertex set is the Cartesian product $V_1 \times V_2$, and
such that two vertices $(u_1,u_2)$ and $(v_1,v_2)$ are adjacent in
$G_1 \ \square \  G_2$ if and only if either $u_1 = v_1$ and $u_2$
is adjacent to $v_2$ in $G_2$, or $u_2 = v_2$ and $u_1$ is
adjacent to $v_1$ in $G_1$.

The following result was proved in~\cite{B-D-thinness}. We include
the proof in order to make some remarks about it. Most of the
proofs for other graph products are structurally similar to this
one.

\begin{theorem}\label{thm:cart}\cite{B-D-thinness}
 Let $G_1$ and $G_2$ be graphs. Then, for $f \in \{\thin,\pthin\}$, {$f(G_1 \ \square \  G_2) \leq f(G_1)|V(G_2)|$}.
\end{theorem}

\begin{proof}
Let $G_1$ be a $k$-thin (resp. proper $k$-thin) graph. Consider
$V_1 \times V_2$ lexicographically ordered with respect to the
{defined} orderings of $V_1$ and $V_2$. Consider now the partition
$\{V^{i,j}\}_{1\leq i \leq k,\ 1 \leq j \leq n_2}$ such that
$V^{i,j} = \{(v,w_j) : v \in V_1^i\}$ for each $1\leq i \leq k$,
$1 \leq j \leq n_2$. We will show that this ordering and partition
of $V_1 \times V_2$ are consistent (resp. strongly consistent).
Let $(v_p,w_i), (v_q,w_j), (v_r,w_{\ell})$ be three vertices
appearing in that ordering in $V_1 \times V_2$.

\emph{Case 1: $p = q = r$.} In this case, the three vertices are
in different classes, so no restriction has to be satisfied.

\emph{Case 2: $p = q < r$.} In this case, $(v_p,w_i)$ and
$(v_q,w_j)$ are in different classes. So suppose $G_1$ is proper
$k$-thin and $(v_q,w_j), (v_r,w_{\ell})$ belong to the same class,
i.e., $j=\ell$. Vertices $(v_p,w_i)$ and $(v_r,w_{\ell})$ are
adjacent in $G_1 \ \square \ G_2$ if and only if $i = \ell$ and
$v_pv_r \in E_1$. But then $(v_p,w_i)=(v_q,w_j)$, a contradiction.

\emph{Case 3: $p < q = r$.} In this case, $(v_q,w_j)$ and
$(v_r,w_{\ell})$ are in different classes. So suppose $G_1$ is
$k$-thin (resp. proper $k$-thin) and $(v_p,w_i), (v_q,w_j)$ belong
to the same class, i.e., $i=j$. Vertices $(v_p,w_i)$ and
$(v_r,w_{\ell})$ are adjacent in $G_1 \ \square \ G_2$ if and only
if $i = \ell$ and $v_pv_r \in E_1$. But then
$(v_r,w_{\ell})=(v_q,w_j)$, a contradiction.

\emph{Case 4: $p < q < r$.} Suppose first $G_1$ is $k$-thin (resp.
proper $k$-thin) and $(v_p,w_i), (v_q,w_j)$ belong to the same
class, i.e., $i=j$ and $v_p$, $v_q$ belong to the same class in
$G_1$. Vertices $(v_p,w_i)$ and $(v_r,w_{\ell})$ are adjacent in
$G_1 \ \square \ G_2$ if and only if $i = \ell$ and $v_pv_r \in
E_1$. But then $j=\ell$ and since the ordering and the partition
are consistent (resp. strongly consistent) in $G_1$, $v_rv_q \in
E_1$ and so $(v_r,w_{\ell})$ and $(v_q,w_j)$ are adjacent in $G_1
\ \square \ G_2$. Now suppose that $G_1$ is proper $k$-thin and
$(v_q,w_j), (v_r,w_{\ell})$ belong to the same class, i.e.,
$j=\ell$. Vertices $(v_p,w_i)$ and $(v_r,w_{\ell})$ are adjacent
in $G_1 \ \square \ G_2$ if and only if $i = \ell$ and $v_pv_r \in
E_1$. But then $i=j$ and since the ordering and the partition are
strongly consistent in $G_1$, $v_pv_q \in E_1$ and so $(v_p,w_i)$
and $(v_q,w_j)$ are adjacent in $G_1 \ \square \ G_2$.
\end{proof}

\begin{remark}\label{rem:cart}
Notice that if the partition of $G_1$ consists of independent sets
(respectively, complete sets), the partition defined for $G_1 \
\square \ G_2$ consists also of independent sets (respectively,
complete sets). So, $f(G_1 \ \square \ G_2) \leq f(G_1)|V(G_2)|$,
for $f \in \{\thin, \pthin, \indthin, \compthin, \indpthin,
\comppthin\}$.
\end{remark}

These results can be strengthened by replacing $|V(G_2)|$ by the
size of the largest connected component of $G_2$, by
Theorem~\ref{thm:union} and since $G \ \square \ (H \cup H') = (G
\ \square \ H) \cup (G \ \square \ H')$.

On the negative side, since $P_r$ has independent proper
thinness~2 (with the order given by the definition of path), but
$P_r \ \square \ P_r = GR_r$, we have the following corollary of
Corollary~\ref{cor:grid}.

\begin{corollary}\label{cor:nbcart}
There is no function $f: \mathbb{R}^2 \to \mathbb{R}$ such that
$\thin(G_1 \ \square \ G_2) \leq f(\indpthin(G_1),\indpthin(G_2))$
for any pair of graphs $G_1$, $G_2$.
\end{corollary}

\begin{lemma}\label{lem:Knsq}
For $n \geq 1$, $\thin(K_n \ \square \ K_n) = n$.
\end{lemma}

\begin{proof} For $n \geq 1$, $K_n \ \square \ K_n$ is $(2n-2)$-regular, and
for any pair of vertices $u$, $v$, $|N(u) \cap N(v)| \leq n-2$. By
Corollary~\ref{cor:k-reg}, $\thin(K_n \ \square \ K_n) \geq n$. By
Theorem~\ref{thm:cart}, {$\thin(K_n \ \square \ K_n) \leq n$}.
\end{proof}

\begin{corollary}\label{cor:nbcart2}
There is no function $f: \mathbb{R}^2 \to \mathbb{R}$ such that
$\thin(G_1 \ \square \ G_2) \leq
f(\comppthin(G_1),\comppthin(G_2))$ for any pair of graphs $G_1$,
$G_2$.
\end{corollary}

\begin{lemma}\label{lem:KnKnn}
For $n \geq 1$, $\thin(K_n \ \square \ K_{n,n}) \geq n-1$.
\end{lemma}

\begin{proof} For $n \geq 1$, $K_n \ \square \ K_{n,n}$ is $(2n-1)$-regular, and
for any pair of vertices $u$, $v$, $|N(u) \cap N(v)| \leq n$. By
Corollary~\ref{cor:k-reg}, $\thin(K_n \ \square \ K_{n,n}) \geq n-
1$.
\end{proof}

\begin{corollary}\label{cor:nbcart3}
There is no function $f: \mathbb{R}^2 \to \mathbb{R}$ such that
$\thin(G_1 \ \square \ G_2) \leq
f(\comppthin(G_1),\indpthin(G_2))$ for any pair of graphs $G_1$,
$G_2$.
\end{corollary}

The non existence of bounds in terms of other parameters can be
deduced from diagram in Figure~\ref{fig:hasse}.

Further consequences of the examples above are the following.

\begin{corollary}\label{cor:lowercart}
Given two connected graphs $G_1$ and $G_2$, $(\min\{\diam(G_1),$
$\diam(G_2)\}+1)/4 \leq (\min\{\lip(G_1),\lip(G_2)\}+1)/4 \leq
\thin(G_1 \ \square \ G_2)$.
\end{corollary}

\begin{corollary}\label{cor:lowercart2}
Given two graphs $G_1$ and $G_2$, $\min\{\omega(G_1),\omega(G_2)\}
\leq $ $\thin(G_1 \ \square \ G_2)$.
\end{corollary}

\subsection{Tensor or direct or categorical product}

The \emph{tensor product} or \emph{direct product} or
\emph{categorical product} or \emph{Kronecker product} $G_1 \times
G_2$ is a graph whose vertex set is the Cartesian product $V_1
\times V_2$, and such that two vertices $(u_1,u_2)$ and
$(v_1,v_2)$ are adjacent in $G_1 \times G_2$ if and only if
 $u_1$ is adjacent to $v_1$ in $G_1$ and $u_2$ is adjacent to
 $v_2$ in $G_2$.

\begin{theorem}\label{thm:direct}
Let $G_1$ and $G_2$ be graphs. Then {$f(G_1 \times G_2) \leq
\indf(G_1 \times G_2) \leq \indf(G_1)|V(G_2)| \leq
f(G_1)\chi(G_1)|V(G_2)|$}, for $f \in \{\thin, \pthin\}$.
\end{theorem}

This result can be strengthened by replacing $|V(G_2)|$ by the
size of the largest connected component of $G_2$, by
Theorem~\ref{thm:union-ind} and since $G \times (H \cup H') = (G
\times H) \cup (G \times H')$.

Since $\comppthin(K_n) = 1$, $|K_2|=2$, and $K_n \times K_2 =
CR_n$, we have the following consequence of
Corollary~\ref{cor:crown}.

\begin{corollary}\label{cor:nbdirect}
There is no function $f: \mathbb{R}^2 \to \mathbb{R}$ such that
$\thin(G_1 \times G_2) \leq f(\comppthin(G_1),|V(G_2)|)$ for any
pair of graphs $G_1$, $G_2$.
\end{corollary}

The graph $P_{2r-1} \times P_{2r-1}$ contains $GR_r$ as an induced
subgraph. Since $\indpthin(P_{2r-1}) = 2$, we have also the
following.

\begin{corollary}\label{cor:nbdirect2}
There is no function $f: \mathbb{R}^2 \to \mathbb{R}$ such that
$\thin(G_1 \times G_2) \leq f(\indpthin(G_1),\indpthin(G_2))$ for
any pair of graphs $G_1$, $G_2$.
\end{corollary}

The non existence of bounds in terms of other parameters can be
deduced from diagram in Figure~\ref{fig:hasse}.

Further consequences of the examples above are the following.

\begin{corollary}\label{cor:lowerdirect}
Given two graphs $G_1$ and $G_2$, if $G_2$ has at least one edge,
then $\omega(G_1)/2 \leq \thin(G_1 \times G_2)$.
\end{corollary}

\begin{corollary}\label{cor:lowerdirect2}
Given two connected graphs $G_1$ and $G_2$, $(\min\{\diam(G_1),$
$\diam(G_2)\}+1)/8 \leq (\min\{\lip(G_1),\lip(G_2)\}+1)/8 \leq
\thin(G_1 \times G_2)$.
\end{corollary}

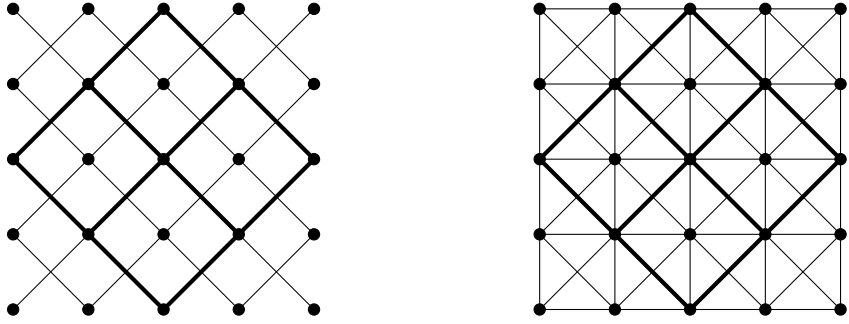
\begin{figure}
    \begin{center}
    \begin{tikzpicture}
    \vertex{0}{0}{v00};
    \vertex{0}{1}{v01};
    \vertex{0}{2}{v02};
    \vertex{0}{3}{v03};
    \vertex{0}{4}{v04};

    \vertex{1}{0}{v10};
    \vertex{1}{1}{v11};
    \vertex{1}{2}{v12};
    \vertex{1}{3}{v13};
    \vertex{1}{4}{v14};

    \vertex{2}{0}{v20};
    \vertex{2}{1}{v21};
    \vertex{2}{2}{v22};
    \vertex{2}{3}{v23};
    \vertex{2}{4}{v24};

    \vertex{3}{0}{v30};
    \vertex{3}{1}{v31};
    \vertex{3}{2}{v32};
    \vertex{3}{3}{v33};
    \vertex{3}{4}{v34};

    \vertex{4}{0}{v40};
    \vertex{4}{1}{v41};
    \vertex{4}{2}{v42};
    \vertex{4}{3}{v43};
    \vertex{4}{4}{v44};

    \edge{v00}{v11};
    \edge{v01}{v12};
    \edge[ultra thick]{v02}{v13};
    \edge{v03}{v14};
    \edge{v01}{v10};
    \edge[ultra thick]{v02}{v11};
    \edge{v03}{v12};
    \edge{v04}{v13};

    \edge{v10}{v21};
    \edge[ultra thick]{v11}{v22};
    \edge{v12}{v23};
    \edge[ultra thick]{v13}{v24};
    \edge[ultra thick]{v11}{v20};
    \edge{v12}{v21};
    \edge[ultra thick]{v13}{v22};
    \edge{v14}{v23};

    \edge[ultra thick]{v20}{v31};
    \edge{v21}{v32};
    \edge[ultra thick]{v22}{v33};
    \edge{v23}{v34};
    \edge{v21}{v30};
    \edge[ultra thick]{v22}{v31};
    \edge{v23}{v32};
    \edge[ultra thick]{v24}{v33};

    \edge{v30}{v41};
    \edge[ultra thick]{v31}{v42};
    \edge{v32}{v43};
    \edge{v33}{v44};
    \edge{v31}{v40};
    \edge{v32}{v41};
    \edge[ultra thick]{v33}{v42};
    \edge{v34}{v43};

    \vertex{7}{0}{w00};
    \vertex{7}{1}{w01};
    \vertex{7}{2}{w02};
    \vertex{7}{3}{w03};
    \vertex{7}{4}{w04};

    \vertex{8}{0}{w10};
    \vertex{8}{1}{w11};
    \vertex{8}{2}{w12};
    \vertex{8}{3}{w13};
    \vertex{8}{4}{w14};

    \vertex{9}{0}{w20};
    \vertex{9}{1}{w21};
    \vertex{9}{2}{w22};
    \vertex{9}{3}{w23};
    \vertex{9}{4}{w24};

    \vertex{10}{0}{w30};
    \vertex{10}{1}{w31};
    \vertex{10}{2}{w32};
    \vertex{10}{3}{w33};
    \vertex{10}{4}{w34};

    \vertex{11}{0}{w40};
    \vertex{11}{1}{w41};
    \vertex{11}{2}{w42};
    \vertex{11}{3}{w43};
    \vertex{11}{4}{w44};

    \edge{w00}{w11};
    \edge{w01}{w12};
    \edge[ultra thick]{w02}{w13};
    \edge{w03}{w14};
    \edge{w01}{w10};
    \edge[ultra thick]{w02}{w11};
    \edge{w03}{w12};
    \edge{w04}{w13};

    \edge{w10}{w21};
    \edge[ultra thick]{w11}{w22};
    \edge{w12}{w23};
    \edge[ultra thick]{w13}{w24};
    \edge[ultra thick]{w11}{w20};
    \edge{w12}{w21};
    \edge[ultra thick]{w13}{w22};
    \edge{w14}{w23};

    \edge[ultra thick]{w20}{w31};
    \edge{w21}{w32};
    \edge[ultra thick]{w22}{w33};
    \edge{w23}{w34};
    \edge{w21}{w30};
    \edge[ultra thick]{w22}{w31};
    \edge{w23}{w32};
    \edge[ultra thick]{w24}{w33};

    \edge{w30}{w41};
    \edge[ultra thick]{w31}{w42};
    \edge{w32}{w43};
    \edge{w33}{w44};
    \edge{w31}{w40};
    \edge{w32}{w41};
    \edge[ultra thick]{w33}{w42};
    \edge{w34}{w43};

    \edge{w00}{w01};
    \edge{w01}{w02};
    \edge{w02}{w03};
    \edge{w03}{w04};

    \edge{w10}{w11};
    \edge{w11}{w12};
    \edge{w12}{w13};
    \edge{w13}{w14};

    \edge{w20}{w21};
    \edge{w21}{w22};
    \edge{w22}{w23};
    \edge{w23}{w24};

    \edge{w30}{w31};
    \edge{w31}{w32};
    \edge{w32}{w33};
    \edge{w33}{w34};

    \edge{w40}{w41};
    \edge{w41}{w42};
    \edge{w42}{w43};
    \edge{w43}{w44};

    \edge{w00}{w10};
    \edge{w10}{w20};
    \edge{w20}{w30};
    \edge{w30}{w40};

    \edge{w01}{w11};
    \edge{w11}{w21};
    \edge{w21}{w31};
    \edge{w31}{w41};

    \edge{w02}{w12};
    \edge{w12}{w22};
    \edge{w22}{w32};
    \edge{w32}{w42};

    \edge{w03}{w13};
    \edge{w13}{w23};
    \edge{w23}{w33};
    \edge{w33}{w43};

    \edge{w04}{w14};
    \edge{w14}{w24};
    \edge{w24}{w34};
    \edge{w34}{w44};
    \end{tikzpicture}
    \end{center}
    \caption{The  $(r \times r)$-grid $GR_r$ as an induced subgraph of  $P_{2r-1} \times P_{2r-1}$ and of  $P_{2r-1} \boxtimes P_{2r-1}$.}
\end{figure}

\subsection{Strong or normal product}

The \emph{strong product} (also known as \emph{normal product})
$G_1 \boxtimes G_2$ is a graph whose vertex set is the Cartesian
product $V_1 \times V_2$, and such that two vertices $(u_1,u_2)$
and $(v_1,v_2)$ are adjacent in $G_1 \boxtimes G_2$ if and only if
they are adjacent either in $G_1 \ \square \  G_2$ or in $G_1
\times G_2$.

\begin{theorem}\label{thm:strong}
Let $G_1$ and $G_2$ be graphs. Then $f(G_1 \boxtimes G_2) \leq
f(G_1)|V(G_2)|$ for $f \in \{\thin, \pthin, \compthin, \comppthin,
\indthin, \indpthin\}$. Moreover, if $G_2$ is complete, then
$f(G_1 \boxtimes G_2) = f(G_1)$ for $f \in \{\thin, \pthin,
\compthin, \comppthin\}$.
\end{theorem}

This result can be strengthened for $f \in \{\thin, \pthin,
\indthin, \indpthin\}$, replacing $|V(G_2)|$ by the size of the
largest connected component of $G_2$, by Theorems~\ref{thm:union}
and~\ref{thm:union-ind}, and since $G \boxtimes (H \cup H') = (G
\boxtimes H) \cup (G \boxtimes H')$.

The graph $P_{2r-1} \boxtimes P_{2r-1}$ contains $GR_r$ as an
induced subgraph. Since $\indpthin(P_{2r-1}) = 2$, we have the
following corollary of Corollary~\ref{cor:grid}.

\begin{corollary}\label{cor:nbstrong}
There is no function $f: \mathbb{R}^2 \to \mathbb{R}$ such that
$\thin(G_1 \boxtimes G_2) \leq f(\indpthin(G_1),\indpthin(G_2))$
for any pair of graphs $G_1$, $G_2$.
\end{corollary}

\begin{lemma}\label{lem:Knsqbox}
For $n \geq 2$, $\thin((K_n \ \square \ K_2) \boxtimes (K_n \
\square \ K_2)) \geq n+2$.
\end{lemma}

\begin{proof} For $n \geq 2$, $(K_n \ \square \ K_2) \boxtimes (K_n \
\square \ K_2)$ is $(n^2+2n)$-regular, and for any pair of
vertices $u$, $v$, $|N(u) \cap N(v)| \leq n^2+n-2$. By
Corollary~\ref{cor:k-reg}, $\thin((K_n \ \square \ K_2) \boxtimes
(K_n \ \square \ K_2)) \geq n+2$.
\end{proof}

Remark~\ref{rem:cart} implies $\comppthin(K_n \ \square \ K_2) =
2$ for $n \geq 2$, since the graph contains an induced cycle of
length four, which is not an interval graph. So we have also the
following.

\begin{corollary}\label{cor:nbstrong2}
There is no function $f: \mathbb{R}^2 \to \mathbb{R}$ such that
$\thin(G_1 \boxtimes G_2) \leq f(\comppthin(G_1),\comppthin(G_2))$
for any pair of graphs $G_1$, $G_2$.
\end{corollary}

The non existence of bounds in terms of other parameters can be
deduced from diagram in Figure~\ref{fig:hasse}.

A further consequence of the example used for
Corollary~\ref{cor:nbstrong} is the following.

\begin{corollary}\label{cor:lowerstrong}
Given two connected graphs $G_1$ and $G_2$, $(\min\{\diam(G_1),$
$\diam(G_2)\}+1)/8 \leq (\min\{\lip(G_1),\lip(G_2)\}+1)/8 \leq
\thin(G_1 \boxtimes G_2)$.
\end{corollary}

\subsection{Co-normal or disjunctive product}

The \emph{co-normal product} or \emph{disjunctive product} $G_1
\ast G_2$ is a graph whose vertex set is the Cartesian product
$V_1 \times V_2$, and such that two vertices $(u_1,u_2)$ and
$(v_1,v_2)$ are adjacent in $G_1 \ast G_2$ if and only if either
$u_1$ is adjacent to $v_1$ in $G_1$ or $u_2$ is adjacent to $v_2$
in $G_2$.

Notice that $\overline{G \ast H} = \overline{G} \boxtimes
\overline{H}$.

\begin{theorem}\label{thm:conorm}
Let $G_1$ and $G_2$ be graphs. Then $f(G_1 \ast G_2) \leq
\indf(G_1 \ast G_2) \leq \indf(G_1)|V(G_2)| \leq
f(G_1)\chi(G_1)|V(G_2)|$, for $f \in \{\thin, \pthin\}$.
\end{theorem}

Since $\comppthin(K_t) = 1$, $|2K_1|=2$, and $K_t \ast 2K_1 =
\overline{tK_2}$, we have the following corollary of
Theorem~\ref{thm:tK2}.

\begin{corollary}\label{cor:nbconorm}
There is no function $f: \mathbb{R}^2 \to \mathbb{R}$ such that
$\thin(G_1 \ast G_2) \leq f(\comppthin(G_1),|V(G_2)|)$ for any
pair of graphs $G_1$, $G_2$.
\end{corollary}

Consider the graph $tK_2 \ast tK_2$. It is $(4t-1)$-regular, and
for every pair of vertices $u,v$, it holds $|N(u) \cap N(v)| \leq
2t+1$. By Corollary~\ref{cor:k-reg}, $\thin(tK_2 \ast tK_2) \geq
2t-2$. Since $\indpthin(tK_2) = 2$, we have also the following
corollary.

\begin{corollary}\label{cor:nbconorm2}
There is no function $f: \mathbb{R}^2 \to \mathbb{R}$ such that
$\thin(G_1 \ast G_2) \leq f(\indpthin(G_1),\indpthin(G_2))$ for
any pair of graphs $G_1$, $G_2$.
\end{corollary}

The non existence of bounds in terms of other parameters can be
deduced from diagram in Figure~\ref{fig:hasse}.

Further consequences of the examples above are the following.

\begin{corollary}\label{cor:lowerconorm}
Given two graphs $G_1$ and $G_2$, if $G_2$ is not complete, then
$\omega(G_1) \leq \thin(G_1 \ast G_2)$.
\end{corollary}

\begin{corollary}\label{cor:lowerconorm2}
Given two graphs $G_1$ and $G_2$, $2\min\{\mim(G_1),\mim(G_2)\}-2
\leq \thin(G_1 \ast G_2)$.
\end{corollary}

\subsection{Modular product}

The \emph{modular product} $G_1 \diamond G_2$ is a graph whose
vertex set is the Cartesian product $V_1 \times V_2$, and such
that two vertices $(u_1,u_2)$ and $(v_1,v_2)$ are adjacent in $G_1
\diamond G_2$ if and only if either $u_1$ is adjacent to $v_1$ in
$G_1$ and $u_2$ is adjacent to $v_2$ in $G_2$, or $u_1$ is
nonadjacent to $v_1$ in $G_1$ and $u_2$ is nonadjacent to $v_2$ in
$G_2$.

Notice that $K_n \diamond K_2 = CR_n$ and $tK_2 \diamond K_1 =
\overline{tK_2}$, so we have the following.

\begin{corollary}\label{cor:nbmodular}
There is no function $f: \mathbb{R}^2 \to \mathbb{R}$ such that
$\thin(G_1 \diamond G_2) \leq f(h(G_1),|G_2|)$, for $h \in
\{\comppthin,\indpthin\}$ for any pair of graphs $G_1$, $G_2$.
\end{corollary}

The non existence of bounds in terms of other parameters can be
deduced from diagram in Figure~\ref{fig:hasse}.

Further consequences of the examples above are the following.

\begin{corollary}\label{cor:lowermodular}
Given two graphs $G_1$ and $G_2$, if $G_2$ has at least one edge,
then $\omega(G_1)/2 \leq \thin(G_1 \diamond G_2)$.
\end{corollary}

\begin{corollary}\label{cor:lowermodular2}
Given two graphs $G_1$ and $G_2$, $\mim(G_1) \leq \thin(G_1
\diamond G_2)$.
\end{corollary}

\subsection{Homomorphic product}

The \emph{homomorphic product} $G_1 \ltimes
G_2$~\cite{M-R-hom-prod} is a graph whose vertex set is the
Cartesian product $V_1 \times V_2$, and such that two vertices
$(u_1,u_2)$ and $(v_1,v_2)$ are adjacent in $G_1 \ltimes G_2$ if
and only if either $u_1 = v_1$ or $u_1$ is adjacent to $v_1$ in
$G_1$ and $u_2$ is nonadjacent to $v_2$ in $G_2$. It is not
necessarily commutative.

\begin{theorem}\label{thm:homo}
Let $G_1$ and $G_2$ be graphs. Then $f(G_1 \ltimes G_2) \leq
f(G_1)|V(G_2)|$ for $f \in \{\thin, \pthin, \compthin, \comppthin,
\indthin, \indpthin\}$.
\end{theorem}

Let $G_1$ be isomorphic to $K_2$, such that $V(G_1) =
\{v_1,v_2\}$, and $G_2$ be isomorphic to $tK_2$,  for some $t \geq
1$, such that $V(G_2) = \{w_1,z_1,\dots,w_t,z_t\}$ and $E(G_2) =
\{w_iz_i : 1 \leq i \leq t\}$. In $G_1 \ltimes G_2$, the vertices
$\{(v_1,w_i) : 1 \leq i \leq t\} \cup \{(v_2,z_i) : 1 \leq i \leq
t\}$ induce the graph $\overline{tK_2}$. In other words, $K_2
\ltimes tK_2$ has $\overline{tK_2}$ as induced subgraph.

Now let $G_3$ be isomorphic to $K_t \ \square \ K_2$, for some $t
\geq 1$, such that $V(G_3) = \{w_1,z_1,\dots,w_t,z_t\}$, $\{w_1,
\dots, w_t\}$ is a clique, $\{z_1, \dots, z_t\}$ is a clique, and
$w_i$ is adjacent to $z_j$ if and only if $i = j$. In $G_1 \ltimes
G_3$, the vertices $\{(v_1,w_i) : 1 \leq i \leq t\} \cup
\{(v_2,z_i) : 1 \leq i \leq t\}$ induce the graph
$\overline{tK_2}$. In other words, $K_2 \ltimes (K_t \ \square \
K_2)$ has $\overline{tK_2}$ as induced subgraph.

Since $\indpthin(tK_2)=2$ and $\comppthin(K_n \ \square \ K_2)=2$
(Remark~\ref{rem:cart}), we have the following.

\begin{corollary}\label{cor:nbhomo}
There is no function $f: \mathbb{R}^2 \to \mathbb{R}$ such that
$\thin(G_1 \ltimes G_2) \leq f(|G_1|,h(G_2))$, for $h \in
\{\indpthin,\comppthin\}$ for any pair of graphs $G_1$, $G_2$.
\end{corollary}

The non existence of bounds in terms of other parameters can be
deduced from diagram in Figure~\ref{fig:hasse}.

Another consequence of the example above is the following.

\begin{corollary}\label{cor:lowerhomo}
Given two graphs $G_1$ and $G_2$, if $G_1$ has at least one edge,
then $\mim(G_2) \leq \thin(G_1 \ltimes G_2)$.
\end{corollary}

\subsection{Hom-product}

The \emph{hom-product} $G_1 \circ G_2$~\cite{Bacik-phd-hom-prod}
is a graph whose vertex set is the Cartesian product $V_1 \times
V_2$, and such that two vertices $(u_1,u_2)$ and $(v_1,v_2)$ are
adjacent in $G_1 \circ G_2$ if and only if $u_1 \neq v_1$ and
either $u_1$ is nonadjacent to $v_1$ in $G_1$ or $u_2$ is adjacent
to $v_2$ in $G_2$. It is not necessarily commutative, indeed, $G_1
\circ G_2 = \overline{G_1 \ltimes G_2}$.

\begin{theorem}\label{thm:hom}
Let $G_1$ and $G_2$ be graphs. Then {$f(G_1 \circ G_2) \leq
\indf(G_1 \circ G_2) \leq |V(G_1)|\indf(G_2)$}, for $f \in
\{\thin, \pthin\}$.
\end{theorem}

\begin{proof}
Let $G_2$ be a $k$-independent-thin (resp. proper
$k$-independent-thin) graph. Consider $V_1 \times V_2$
lexicographically ordered with respect to the {defined} orderings
of $V_2$ and $V_1$. Consider now the partition $\{V^{i,j}\}_{1\leq
i \leq n_1,\ 1 \leq j \leq k}$ such that $V^{i,j} = \{(v_i,w) : w
\in V_2^j\}$ for each $1\leq i \leq n_1$, $1 \leq j \leq k$. By
definition of the hom-product, each $V^{i,j}$ is an independent
set.

We will show now that this ordering and partition of $V_1 \times
V_2$ are consistent (resp. strongly consistent). Let $(v_p,w_i),
(v_q,w_j), (v_r,w_{\ell})$ be three vertices appearing in that
ordering in $V_1 \times V_2$.

\emph{Case 1: $i = j = \ell$.} In this case, the three vertices
are in different classes, so no restriction has to be satisfied.

\emph{Case 2: $i = j < \ell$.} In this case, $(v_p,w_i)$ and
$(v_q,w_j)$ are in different classes. So, suppose $G_2$ is proper
$k$-independent-thin and $(v_q,w_j), (v_r,w_{\ell})$ belong to the
same class, i.e., $q=r$ and $w_i = w_j$ and $w_{\ell}$ belong to
the same class in $G_2$. In particular, since the classes are
independent sets, $w_i w_{\ell} \not \in E_2$. Vertices
$(v_p,w_i)$ and $(v_r,w_{\ell})$ are adjacent in $G_1 \circ G_2$
if and only if $p\neq r$ and either $v_p v_r \not \in E_1$ or $w_i
w_{\ell} \in E_2$. So, assume $p \neq r$. Since $q=r$, $p \neq q$,
and since the graph is loopless and $i = j$, $w_i w_j \not \in
E_2$. So $(v_p,w_i)$ is adjacent to $(v_q,w_j)$ in $G_1 \circ
G_2$, as required.

\emph{Case 3: $i < j = \ell$.} In this case, $(v_q,w_j)$ and
$(v_r,w_{\ell})$ are in different classes. So suppose $G_2$ is
$k$-independent-thin (resp. proper $k$-independent-thin) and
$(v_p,w_i), (v_q,w_j)$ belong to the same class, i.e., $p=q$ and
$w_i$ and $w_j = w_{\ell}$ belong to the same class in $G_2$. In
particular, since the classes are independent sets, $w_i w_j \not
\in E_2$. Vertices $(v_p,w_i)$ and $(v_r,w_{\ell})$ are adjacent
in $G_1 \circ G_2$ if and only if $p \neq r$ and either $v_p v_r
\not \in E_1$ or $w_i w_{\ell} \in E_2$. So, assume $p \neq r$.
Since $p=q$, $q \neq r$, and since the graph is loopless and $j =
\ell$, $w_j w_{\ell} \not \in E_2$. So $(v_q,w_j)$ is adjacent to
$(v_r,w_{\ell})$ in $G_1 \circ G_2$, as required.

\emph{Case 4: $i < j < \ell$.} Suppose first that $G_2$ is
$k$-independent-thin (resp. proper $k$-independent-thin) and
$(v_p,w_i), (v_q,w_j)$ belong to the same class, i.e., $p=q$ and
$w_i$, $w_j$ belong to the same class in $G_2$. Since the classes
are independent, $w_i$ and $w_j$ are not adjacent. Vertices
$(v_p,w_i)$ and $(v_r,w_{\ell})$ are adjacent in $G_1 \circ G_2$
if and only if $p \neq r$ and either $v_p v_r \not \in E_1$ or
$w_i w_{\ell} \in E_2$. So, assume $p \neq r$, and since $p=q$, $q
\neq r$ too. If $v_p v_r \not \in E_1$, then since $p=q$, $v_q v_r
\not \in E_1$, and $(v_q,w_j)$ is adjacent to $(v_r,w_{\ell})$ in
$G_1 \circ G_2$, as required. If $w_i w_{\ell} \in E_2$, since
$w_i$, $w_j$ belong to the same class in $G_2$ and the partition
of $V_2$ is (strongly) consistent with the ordering, {$w_j
w_{\ell} \in E_2$}, and $(v_q,w_j)$ is adjacent to
$(v_r,w_{\ell})$ in $G_1 \circ G_2$, as required.

Now suppose that $G_1$ is proper $k$-thin and $(v_q,w_j),
(v_r,w_{\ell})$ belong to the same class, i.e., $q = r$ and $w_j$,
$w_{\ell}$ belong to the same class in $G_2$. Vertices $(v_p,w_i)$
and $(v_r,w_{\ell})$ are adjacent in $G_1 \circ G_2$ if and only
if $p \neq r$ and either $v_p v_r \not \in E_1$ or $w_i w_{\ell}
\in E_2$. So, assume $p \neq r$, and since $q = r$, $p \neq q$
too. If $v_p v_r \not \in E_1$, then since $q = r$, $v_p v_q \not
\in E_1$, and $(v_p,w_i)$ and $(v_q,w_j)$ are adjacent in $G_1
\circ G_2$, as required. If $w_i w_{\ell} \in E_2$, since $w_j$,
$w_{\ell}$ belong to the same class in $G_2$ and the partition of
$V_2$ is strongly consistent with the ordering, $w_i w_j \in E_2$,
and $(v_p,w_i)$ is adjacent to $(v_q,w_j)$ in $G_1 \circ G_2$, as
required.
\end{proof}

Notice that $K_2 \circ K_n = CR_n$. Also, $G \circ K_1 =
\overline{G}$, so, by taking $G = tK_2$ and $G = K_n \ \square \
K_2$, we have the following.

\begin{corollary}\label{cor:nbhom}
There is no function $f: \mathbb{R}^2 \to \mathbb{R}$ such that
$\thin(G_1 \circ G_2) \leq f(|G_1|,\comppthin(G_2))$, and there is
no function $f: \mathbb{R}^2 \to \mathbb{R}$ such that $\thin(G_1
\circ G_2) \leq f(h(G_1),|G_2|)$ for $h \in
\{\indpthin,\comppthin\}$ for any pair of graphs $G_1$, $G_2$.
\end{corollary}

The non existence of bounds in terms of other parameters can be
deduced from diagram in Figure~\ref{fig:hasse}.

Further consequences of the examples above are the following.

\begin{corollary}\label{cor:lowerhom}
Given two graphs $G_1$ and $G_2$, if $G_1$ has at least one edge,
then $\omega(G_2)/2 \leq \thin(G_1 \circ G_2)$.
\end{corollary}

\begin{corollary}\label{cor:lowerhom2}
Given two graphs $G_1$ and $G_2$, $\mim(G_1) \leq \thin(G_1 \circ
G_2)$.
\end{corollary}

\section{Conclusion}\label{sec:conc}

In this paper, we give upper bounds for the thinness, complete
thinness, independent thinness and their proper versions for the
union and join of graphs, as well as the lexicographical,
Cartesian, direct, strong, disjunctive, modular, homomorphic and
hom-products of graphs. These bounds are given in terms of the
parameters (depicted in the Hasse diagram of
Figure~\ref{fig:hasse}) of the component graphs, and for each of
the cases, it is proved that no upper bound in terms of lower
parameters (from those in the diagram) of the component graphs
exists. The non existence proofs are based on the determination of
exact values or lower bounds for the thinness of some families of
graphs like complements of matchings, grids, crown graphs,
hypercubes, or products of simple graphs like complete graphs,
stable sets, induced matchings and induced paths. We summarize the
main bounds obtained and the graph families with high thinness in
Table~\ref{table1}.

\begin{table}
\footnotesize
\begin{tabular}{|l|l|}
  \hline
  \textbf{Upper bounds} & \textbf{High thinness}   \\

  \hline


 $\opthin(G_1 \cup G_2) =\max\{\opthin(G_1),\opthin(G_2)\}$ &    \\

 $\opindthin(G_1 \cup G_2) =\max\{\opindthin(G_1),\opindthin(G_2)\}$ &   \\

 $\opcompthin(G_1 \cup G_2) =\opcompthin(G_1)+\opcompthin(G_2)$ &   \\  \hline


 $\thin(G_1^* \vee G_2^*)=\thin(G_1^*)+\thin(G_2^*)$ &    \\

 $\thin(G_1 \vee K_n)=\thin(G_1)$ &    \\

  $\pthin(G_1 \vee G_2) \leq \pthin(G_1)+\pthin(G_2)$ &    \\

 $\opindthin(G_1 \vee G_2)=\opindthin(G_1)+\opindthin(G_2)$ &   \\

 $\opcompthin(G_1 \vee G_2) \leq \opcompthin(G_1)+\opcompthin(G_2)$ &    \\

 $\compthin(G_1 \vee K_n) =\compthin(G_1)$ &    \\  \hline


  $\opthin(G_1 \bullet G_2) \leq \opindthin(G_1)\cdot \opthin(G_2)$ & $K_n \bullet 2K_1$  \\

 $\opthin(G_1 \bullet K_n) =\opthin(G_1)$ &    \\

 $\opindthin(G_1 \bullet G_2) \leq \opindthin(G_1)\cdot \opindthin(G_2)$ &    \\

 $\opcompthin(G_1 \bullet G_2) \leq |V(G_1)|\cdot \opcompthin(G_2)$ &    \\

 $\opcompthin(G_1 \bullet K_n) =\opcompthin(G_1)$ &    \\  \hline


  $\opthin(G_1 \ \square \ G_2) \leq \opthin(G_1) \cdot |V(G_2)|$ & $P_n \ \square \ P_n$  \\

  $\opindthin(G_1 \ \square \ G_2) \leq \opindthin(G_1) \cdot |V(G_2)|$ & $K_n \ \square \ K_n$  \\

  $\opcompthin(G_1 \ \square \ G_2) \leq \opcompthin(G_1) \cdot |V(G_2)|$ & $K_n \ \square \ K_{n,n}$  \\  \hline


$\opthin(G_1 \times G_2) \leq \opindthin(G_1) \cdot |V(G_2)|$ & $K_n \times K_2$ \\

$\opindthin(G_1 \times G_2) \leq \opindthin(G_1) \cdot |V(G_2)|$ &
$P_n \times P_n$
\\ \hline


$\opthin(G_1 \boxtimes G_2) \leq \opthin(G_1)\cdot |V(G_2)|$ & $P_n \boxtimes P_n$  \\

$\opthin(G_1 \boxtimes K_n) = \opthin(G_1)$ & $(K_n \ \square \ K_2) \boxtimes (K_n \ \square \ K_2)$  \\

$\opindthin(G_1 \boxtimes G_2) \leq \opindthin(G_1)\cdot |V(G_2)|$ &    \\

$\opcompthin(G_1 \boxtimes G_2) \leq \opcompthin(G_1)\cdot |V(G_2)|$ &    \\

$\opcompthin(G_1 \boxtimes G_2) = \opcompthin(G_1)$ &    \\ \hline


$\opthin(G_1 \ast G_2) \leq \opindthin(G_1) \cdot |V(G_2)|$ & $K_n \ast 2K_1$ \\

$\opindthin(G_1 \ast G_2) \leq \opindthin(G_1) \cdot |V(G_2)|$ &
$nK_2 \ast nK_2$
\\ \hline


& $K_n \diamond K_2$ \\

& $nK_2 \diamond K_1$ \\ \hline


$\opthin(G_1 \ltimes G_2) \leq \opthin(G_1)\cdot |V(G_2)|$ & $K_2
\ltimes nK_2$ \\

$\opindthin(G_1 \ltimes G_2) \leq \opindthin(G_1)\cdot |V(G_2)|$ &
$K_2
\ltimes (K_n \ \square \ K_2)$ \\

$\opcompthin(G_1 \ltimes G_2) \leq \opcompthin(G_1)\cdot |V(G_2)|$ & \\
\hline


$\opthin(G_1 \circ G_2) \leq |V(G_1)| \cdot \opindthin(G_2)$ & $K_2 \circ K_n$ \\

$\opindthin(G_1 \circ G_2) \leq |V(G_1)| \cdot \opindthin(G_2)$  & $nK_2 \circ K_1$ \\

& $(K_n \ \square \ K_2) \circ K_1$ \\ \hline

\end{tabular}

\caption{We summarize the upper bounds (when needed, graphs with
an asterisk are not complete). We also summarize the families of
graphs with bounded parameters whose product have high thinness,
used to show the nonexistence of bounds in terms of certain
parameters. Recall that the lexicographic, homomorphic, and
hom-product are not necessarily commutative.}\label{table1}
\end{table}

Furthermore, we describe new general lower and upper bounds for
the thinness of graphs, and some lower bounds for the graph
operations in terms of other well known graph invariants like
clique number, maximum induced matching, longest induced path, or
diameter.

Some open problems and possible research directions are:

\begin{itemize}
\item It would be interesting to find tighter bounds, in the case
in which it is possible.

 \item It remains as an open problem the computational
complexity of computing the independent and complete (proper)
thinness for general graphs. For co-comparability graphs, both the
independent thinness and independent proper thinness are exactly
the chromatic number, which can be computed in polynomial
time~\cite{Go-comp2}.

\item Regarding lower and upper bounds in
Section~\ref{sec:families}, can we have some similar results for
proper thinness? Or for the independent and complete versions of
(proper) thinness?

\item Does there exist a graph $G$ such that $\thin(G) >
|V(G)|/2$?

\item Does there exist a co-comparability graph $G$ such that
$\pthin(G) > |V(G)|/2$?
\end{itemize}

\section*{Acknowledgements}
\label{sec:ack}

This work was done when Moys\'es S. Sampaio Jr. was visiting the
University of Buenos Aires, funded by a grant from FAPERJ. The
work was also partially supported by UBACyT Grants
20020170100495BA and 20020160100095BA (Argentina), FAPERJ, CAPES
and CNPq (Brazil), and Programa Regional MATHAMSUD MATH190013.
Carolina L. Gonzalez is partially supported by a CONICET doctoral
fellowship.

We want to thank the anonymous referees for their valuable
suggestions that helped us improving this work.



\newpage
\appendix
\section{Omitted Proofs of Section~\ref{sec:thin-and-oper}}\label{apdx:sec4:proofs}

{The operations involved in these theorems are defined over a pair
of graphs $G_1 = (V_1,E_1)$ and $G_2 = (V_2,E_2)$ such that $|V_1|
= n_1$, $|V_2| = n_2$ and $V_1 \cap V_2 = \emptyset$. Besides, for
some of the following proofs, we consider an implicit ordering and
partition for both $V_1$ and $V_2$, as defined next. The ordering
of $V_1$ will be denoted by $v_1,\dots, v_{n_1}$ and that of $V_2$
by $w_1,\dots, w_{n_2}$. Moreover, if the value $t_i$ of some
variation of thinness of $G_i$ (for $i \in \{1,2\}$) is involved
in the bound to be proved, the implicit ordering is one
consistent, according to the specified variation of thinness, with
a partition $(V_i^1, \ldots, V_i^{t_i})$. If, otherwise, only the
cardinality $n_i$ of $V_i$ is involved in the bound, the implicit
ordering is an arbitrary one.
For instance, if $G_1$ is a proper $t_1$-independent-thin graph, and $t_1$ is involved in the bound to be proved, it means that the implicit ordering and partition of $V_1$ are strongly consistent and all the $t_1$ parts of the partition are independent sets.}\\

\textbf{Theorem~\ref{thm:join-comp}.}\emph{ Let $G_1$ and $G_2$ be
graphs. Then $f(G_1 \vee G_2) \leq f(G_1)+f(G_2)$, for $f \in
\{\compthin, \comppthin\}$. Moreover, if $G_2$ is complete, then
\linebreak $\compthin(G_1 \vee G_2) =
\compthin(G_1)$.}\\

\begin{proof}
Let $G_1$ and $G_2$ be two graphs with complete thinness (resp.
complete proper thinness) $t_1$ and $t_2$, respectively.

For $G = G_1 \vee G_2$, define a partition with $t_1+t_2$ complete
sets as the union of the two partitions, and $v_1,\dots,
v_{n_1},w_1,\dots, w_{n_2}$ as an ordering of $V(G)$.

Let $x,y,z$ be three vertices of $V(G)$ such that $x < y < z$, $xz
\in E(G)$, and $x$ and $y$ are in the same class of the partition
of $V(G)$. Then, in particular, $x$ and $y$ both belong either to
$V_1$ or to $V_2$. If $z$ belongs to the same graph, then $yz \in
E(G)$ because the ordering and partition restricted to each of
$G_1$ and $G_2$ are consistent. Otherwise, $z$ is also adjacent to
$y$ by the definition of join.

We have proved that the defined partition and ordering are
consistent, and thus that $\compthin(G_1 \vee G_2) \leq
\compthin(G_1)+\compthin(G_2)$. The proof of the strong
consistency, given the strong consistency of the partition and
ordering of each of $G_1$ and $G_2$, is symmetric and implies
$\comppthin(G_1 \vee G_2) \leq \comppthin(G_1)+\comppthin(G_2)$.

Suppose now that $G_2$ is complete (in particular, $t_2 = 1$).
Since $G_1$ is an induced subgraph of $G_1 \vee G_2$, then
$\compthin(G_1 \vee G_2) \geq \compthin(G_1)$. For $G = G_1 \vee
G_2$, define a partition $V^1,\dots, V^{t_1}$ such that $V^1 =
V_1^1 \cup V_2^1$ and $V^i = V_1^i$ for $i = 2, \dots, t_1$, and
define $v_1,\dots, v_{n_1},w_1,\dots, w_{n_2}$ as an ordering of
the vertices. By definition of join, $V^1$ is a complete set of
$G$, as well as $V^2, \dots, V^{t_1}$.

Let $x,y,z$ be three vertices of $V(G)$ such that $x < y < z$, $xz
\in E(G)$, and $x$ and $y$ are in the same class of the partition
of $V(G)$. If $z$ belongs to $V_2$, then $z$ is also adjacent to
$y$, because it is adjacent to every vertex in $G-z$. If $z$
belongs to $V_1$, then $x$, $y$, and $z$, belong to $V_1$ due to
the definition of the order of the vertices, and thus $yz \in
E(G)$ because the ordering and partition restricted to $G_1$ are
consistent. This proves $\compthin(G_1 \vee G_2) \leq
\compthin(G_1)$, thus in this case $\compthin(G_1 \vee G_2) =
\compthin(G_1)$.
\end{proof}

\textbf{Theorem~\ref{thm:direct}.}\emph{ Let $G_1$ and $G_2$ be
graphs. Then {$f(G_1 \times G_2) \leq \indf(G_1 \times G_2) \leq
\indf(G_1)|V(G_2)| \leq
f(G_1)\chi(G_1)|V(G_2)|$}, for $f \in \{\thin, \pthin\}$.}\\

\begin{proof}
Let $G_1$ be a $k$-independent-thin (resp. proper
$k$-independent-thin) graph.

Consider $V_1\times V_2$ lexicographically ordered with respect to
the {defined} orderings of $V_1$ and $V_2$. Consider now the
partition $\{V^{i,j}\}_{1\leq i \leq k,\ 1 \leq j \leq n_2}$ such
that $V^{i,j} = \{(v,w_j) : v \in V_1^i\}$ for each $1\leq i \leq
k$, $1 \leq j \leq n_2$. Since the graphs considered here are
loopless, each $V^{i,j}$ is an independent set. We will show now
that this ordering and partition of $V_1 \times V_2$ are
consistent (resp. strongly consistent). Let $(v_p,w_i), (v_q,w_j),
(v_r,w_{\ell})$ be three vertices appearing in that ordering in
$V_1 \times V_2$.

\emph{Case 1: $p = q = r$.} In this case, the three vertices are
in different classes, so no restriction has to be satisfied.

\emph{Case 2: $p = q < r$.} In this case, $(v_p,w_i)$ and
$(v_q,w_j)$ are in different classes. So, suppose $G_1$ is proper
$k$-independent-thin and $(v_q,w_j), (v_r,w_{\ell})$ belong to the
same class, i.e., $j=\ell$ and $v_q = v_p$ and $v_r$ belong to the
same class in $G_1$. In particular, since the classes are
independent sets, $v_p v_r \not \in E_1$. Vertices $(v_p,w_i)$ and
$(v_r,w_{\ell})$ are adjacent in $G_1 \times G_2$ if and only if
$w_i w_{\ell} \in E_2$ and $v_p v_r \in E_1$, a contradiction.

\emph{Case 3: $p < q = r$.} In this case, $(v_q,w_j)$ and
$(v_r,w_{\ell})$ are in different classes. So suppose $G_1$ is
$k$-independent-thin (resp. proper $k$-independent-thin) and
$(v_p,w_i), (v_q,w_j)$ belong to the same class, i.e., $i=j$ and
$v_p$ and $v_q = v_r$ belong to the same class in $G_1$. In
particular, since the classes are independent sets, $v_p v_r \not
\in E_1$. Vertices $(v_p,w_i)$ and $(v_r,w_{\ell})$ are adjacent
in $G_1 \times G_2$ if and only if $w_i w_{\ell} \in E_2$ and $v_p
v_r \in E_1$, a contradiction.

\emph{Case 4: $p < q < r$.} Suppose first $G_1$ is
$k$-independent-thin (resp. proper $k$-independent-thin) and
$(v_p,w_i), (v_q,w_j)$ belong to the same class, i.e., $i=j$ and
$v_p$, $v_q$ belong to the same class in $G_1$. Vertices
$(v_p,w_i)$ and $(v_r,w_{\ell})$ are adjacent in $G_1 \times G_2$
if and only if $w_i w_{\ell} \in E_2$ and $v_p v_r \in E_1$. Since
the ordering and the partition are consistent (resp. strongly
consistent) in $G_1$, $v_rv_q \in E_1$ and so $(v_r,w_{\ell})$ and
$(v_q,w_j)$ are adjacent in $G_1 \times G_2$, as required. Now
suppose that $G_1$ is proper $k$-thin and $(v_q,w_j),
(v_r,w_{\ell})$ belong to the same class, i.e., $j=\ell$ and $v_q,
v_r$ belong to the same class in $G_1$. Since the ordering and the
partition are strongly consistent in $G_1$, if $(v_p,w_i)$ and
$(v_r,w_{\ell})$ are adjacent in $G_1 \times G_2$ then $v_pv_q \in
E_1$ and so $(v_p,w_i)$ and $(v_q,w_j)$ are adjacent in $G_1
\times G_2$, as required.

The last inequality is a consequence of
Remark~\ref{rem:co-comp-ind-thin}.
\end{proof}

\textbf{Theorem~\ref{thm:strong}.}\emph{ Let $G_1$ and $G_2$ be
graphs. Then $f(G_1 \boxtimes G_2) \leq f(G_1)|V(G_2)|$ for $f \in
\{\thin, \pthin, \compthin, \comppthin, \indthin, \indpthin\}$.
Moreover, if $G_2$ is complete, then $f(G_1 \boxtimes G_2) =
f(G_1)$ for $f \in \{\thin, \pthin,
\compthin, \comppthin\}$.}\\

\begin{proof}
If $G_2$ is complete then $G_1 \boxtimes G_2 = G_1 \bullet G_2$,
so by Theorem~\ref{thm:lex}, $\thin(G_1 \boxtimes G_2) =
\thin(G_1)$, $\pthin(G_1 \boxtimes G_2) = \pthin(G_1)$,
$\compthin(G_1 \boxtimes G_2) = \compthin(G_1)$, and
$\comppthin(G_1 \boxtimes G_2) = \comppthin(G_1)$.

Let $G_1$ be a $k$-thin (resp. proper $k$-thin) graph. Consider
$V_1 \times V_2$ lexicographically ordered with respect to the
{defined} orderings of $V_1$ and $V_2$. Consider now the partition
$\{V^{i,j}\}_{1\leq i \leq k,\ 1 \leq j \leq n_2}$ such that
$V^{i,j} = \{(v,w_j) : v \in V_1^i\}$ for each $1\leq i \leq k$,
$1 \leq j \leq n_2$.

If the sets in the partition of $V_1$ are furthermore complete or
independent, so are the sets of the partition of $V_1 \times V_2$.

We will show now that this ordering and partition of $V_1 \times
V_2$ are consistent (resp. strongly consistent). Let $(v_p,w_i),
(v_q,w_j), (v_r,w_{\ell})$ be three vertices appearing in that
ordering in $V_1 \times V_2$.

\emph{Case 1: $p = q = r$.} In this case, the three vertices are
in different classes, so no restriction has to be satisfied.

\emph{Case 2: $p = q < r$.} In this case, $(v_p,w_i)$ and
$(v_q,w_j)$ are in different classes. So suppose $G_1$ is proper
$k$-thin and $(v_q,w_j), (v_r,w_{\ell})$ belong to the same class,
i.e., $j=\ell$. Vertices $(v_p,w_i)$ and $(v_r,w_{\ell})$ are
adjacent in $G_1 \boxtimes G_2$ if and only if either $i = \ell$
and $v_pv_r \in E_1$, or $i \neq \ell$, $w_i w_{\ell} \in E_2$,
and $v_p v_r \in E_1$. In the first case, $(v_p,w_i)=(v_q,w_j)$, a
contradiction. In the second case, since $p=q$ and $j=\ell$,
$(v_q,w_j)$ and $(v_r,w_{\ell})$ are adjacent, as required.

\emph{Case 3: $p < q = r$.} In this case, $(v_q,w_j)$ and
$(v_r,w_{\ell})$ are in different classes. So suppose $G_1$ is
$k$-thin (resp. proper $k$-thin) and $(v_p,w_i), (v_q,w_j)$ belong
to the same class, i.e., $i=j$. Vertices $(v_p,w_i)$ and
$(v_r,w_{\ell})$ are adjacent in $G_1 \boxtimes G_2$ if and only
if either $i = \ell$ and $v_pv_r \in E_1$, or $i \neq \ell$, $w_i
w_{\ell} \in E_2$, and $v_p v_r \in E_1$. In the first case,
$(v_r,w_{\ell})=(v_q,w_j)$, a contradiction. In the second case,
since $q=r$ and $i=j$, $(v_q,w_j)$ and {$(v_r,w_{\ell})$} are
adjacent, as required.

\emph{Case 4: $p < q < r$.} Suppose first that $G_1$ is $k$-thin
(resp. proper $k$-thin) and $(v_p,w_i), (v_q,w_j)$ belong to the
same class, i.e., $i=j$ and $v_p$, $v_q$ belong to the same class
in $G_1$. Vertices $(v_p,w_i)$ and $(v_r,w_{\ell})$ are adjacent
in $G_1 \boxtimes G_2$ if and only if either $i = \ell$ and
$v_pv_r \in E_1$, or $i \neq \ell$, $w_i w_{\ell} \in E_2$, and
$v_p v_r \in E_1$. In the first case, $j=\ell$ and since the
ordering and the partition are consistent (resp. strongly
consistent) in $G_1$, $v_rv_q \in E_1$ and so $(v_r,w_{\ell})$ and
$(v_q,w_j)$ are adjacent in $G_1 \boxtimes G_2$. In the second
case, since $v_p$, $v_q$ belong to the same class in $G_1$ and the
order and the partition are consistent, $v_q v_r \in E_1$. Since
$i = j$, $w_j w_{\ell} \in E_2$. So $(v_r,w_{\ell})$ and
$(v_q,w_j)$ are adjacent in $G_1 \boxtimes G_2$, as required.

Now suppose that $G_1$ is proper $k$-thin and $(v_q,w_j),
(v_r,w_{\ell})$ belong to the same class, i.e., $j=\ell$ and
$v_q$, $v_r$ belong to the same class in $G_1$. Vertices
$(v_p,w_i)$ and $(v_r,w_{\ell})$ are adjacent in $G_1 \boxtimes
G_2$ if and only if either $i = \ell$ and $v_pv_r \in E_1$, or $i
\neq \ell$, $w_i w_{\ell} \in E_2$, and $v_p v_r \in E_1$. In the
first case, $i=j$ and since the ordering and the partition are
strongly consistent in $G_1$, $v_pv_q \in E_1$ and so $(v_p,w_i)$
and $(v_q,w_j)$ are adjacent in $G_1 \boxtimes G_2$. In the second
case, since $v_q$, $v_r$ belong to the same class in $G_1$ and the
order and the partition are strongly consistent, $v_p v_q \in
E_1$. Since $j=\ell$, $w_i w_j \in E_2$. So $(v_p,w_i)$ and
$(v_q,w_j)$ are adjacent in $G_1 \boxtimes G_2$, as required.
\end{proof}

\textbf{Theorem~\ref{thm:conorm}.}\emph{ Let $G_1$ and $G_2$ be
graphs. Then $f(G_1 \ast G_2) \leq \indf(G_1 \ast G_2) \leq
\indf(G_1)|V(G_2)| \leq
f(G_1)\chi(G_1)|V(G_2)|$, for $f \in \{\thin, \pthin\}$.}\\

\begin{proof}
Let $G_1$ be a $k$-independent-thin (resp. proper
$k$-independent-thin) graph. Consider $V_1 \times V_2$
lexicographically ordered with respect to the {defined} orderings
of $V_1$ and $V_2$. Consider now the partition $\{V^{i,j}\}_{1\leq
i \leq k,\ 1 \leq j \leq n_2}$ such that $V^{i,j} = \{(v,w_j) : v
\in V_1^i\}$ for each $1\leq i \leq k$, $1 \leq j \leq n_2$. Since
the graphs considered here are loopless, each $V^{i,j}$ is an
independent set.

We will show now that this ordering and partition of $V_1 \times
V_2$ are consistent (resp. strongly consistent). Let $(v_p,w_i),
(v_q,w_j), (v_r,w_{\ell})$ be three vertices appearing in that
ordering in $V_1 \times V_2$.

\emph{Case 1: $p = q = r$.} In this case, the three vertices are
in different classes, so no restriction has to be satisfied.

\emph{Case 2: $p = q < r$.} In this case, $(v_p,w_i)$ and
$(v_q,w_j)$ are in different classes. So suppose $G_1$ is proper
$k$-independent-thin and $(v_q,w_j), (v_r,w_{\ell})$ belong to the
same class, i.e., $j=\ell$ and $v_q = v_p$ and $v_r$ belong to the
same class in $G_1$. In particular, since the classes are
independent sets, $v_p v_r \not \in E_1$. Vertices $(v_p,w_i)$ and
$(v_r,w_{\ell})$ are adjacent in $G_1 \ast G_2$ if and only if
either $v_p v_r \in E_1$ or $w_i w_{\ell} \in E_2$, so, in this
case, if and only if $w_i w_{\ell} \in E_2$. Since $j = \ell$,
$(v_p,w_i)$ is adjacent to $(v_q,w_j)$ in $G_1 \ast G_2$, as
required.

\emph{Case 3: $p < q = r$.} In this case, $(v_q,w_j)$ and
$(v_r,w_{\ell})$ are in different classes. So suppose $G_1$ is
$k$-independent-thin (resp. proper $k$-independent-thin) and
$(v_p,w_i), (v_q,w_j)$ belong to the same class, i.e., $i=j$ and
$v_p$ and $v_q = v_r$ belong to the same class in $G_1$. In
particular, since the classes are independent sets, $v_p v_r \not
\in E_1$. Vertices $(v_p,w_i)$ and $(v_r,w_{\ell})$ are adjacent
in $G_1 \ast G_2$ if and only if either $v_p v_r \in E_1$ or $w_i
w_{\ell} \in E_2$, so, in this case, if and only if $w_i w_{\ell}
\in E_2$. Since $i = j$, $(v_r,w_{\ell})$ is adjacent to
$(v_q,w_j)$ in $G_1 \ast G_2$, as required.

\emph{Case 4: $p < q < r$.} Suppose first that $G_1$ is
$k$-independent-thin (resp. proper $k$-independent-thin) and
$(v_p,w_i), (v_q,w_j)$ belong to the same class, i.e., $i=j$ and
$v_p$, $v_q$ belong to the same class in $G_1$. Vertices
$(v_p,w_i)$ and $(v_r,w_{\ell})$ are adjacent in $G_1 \ast G_2$ if
and only if either $v_p v_r \in E_1$ or $w_i w_{\ell} \in E_2$. In
the first case, since the ordering and the partition are
consistent (resp. strongly consistent) in $G_1$, $v_rv_q \in E_1$
and so $(v_r,w_{\ell})$ and $(v_q,w_j)$ are adjacent in $G_1 \ast
G_2$, as required. In the second case, since $i = j$,
$(v_r,w_{\ell})$ is adjacent to $(v_q,w_j)$ in $G_1 \ast G_2$, as
required.

Now suppose that $G_1$ is proper $k$-thin and $(v_q,w_j),
(v_r,w_{\ell})$ belong to the same class, i.e., $j=\ell$ and
$v_q$, $v_r$ belong to the same class in $G_1$.

Vertices $(v_p,w_i)$ and $(v_r,w_{\ell})$ are adjacent in $G_1
\ast G_2$ if and only if either $v_p v_r \in E_1$ or $w_i w_{\ell}
\in E_2$. In the first case, since the ordering and the partition
are strongly consistent in $G_1$, $v_pv_q \in E_1$ and so
$(v_p,w_i)$ and $(v_q,w_j)$ are adjacent in $G_1 \ast G_2$, as
required. In the second case, since $j = \ell$, $(v_p,w_i)$ is
adjacent to $(v_q,w_j)$ in $G_1 \ast G_2$, as required.

The last inequality is a consequence of
Remark~\ref{rem:co-comp-ind-thin}.
\end{proof}

\textbf{Theorem~\ref{thm:homo}.}\emph{ Let $G_1$ and $G_2$ be
graphs. Then $f(G_1 \ltimes G_2) \leq f(G_1)|V(G_2)|$ for $f \in
\{\thin, \pthin, \compthin, \comppthin,
\indthin, \indpthin\}$.}\\

\begin{proof} Let $G_1$ be a $k$-thin (resp. proper $k$-thin)
graph. Consider $V_1 \times V_2$ lexicographically ordered with
respect to the {defined} orderings of $V_1$ and $V_2$. Consider
now the partition $\{V^{i,j}\}_{1\leq i \leq k,\ 1 \leq j \leq
n_2}$ such that $V^{i,j} = \{(v,w_j) : v \in V_1^i\}$ for each
$1\leq i \leq k$, $1 \leq j \leq n_2$.

If the sets in the partition of $V_1$ are furthermore complete or
independent, the sets of the partition of $V_1 \times V_2$ are so,
because $G_2$ is loopless.

We will show now that this ordering and partition of $V_1 \times
V_2$ are consistent (resp. strongly consistent). Let $(v_p,w_i),
(v_q,w_j), (v_r,w_{\ell})$ be three vertices appearing in that
ordering in $V_1 \times V_2$.

\emph{Case 1: $p = q = r$.} In this case, the three vertices are
in different classes, so no restriction has to be satisfied.

\emph{Case 2: $p = q < r$.} In this case, $(v_p,w_i)$ and
$(v_q,w_j)$ are in different classes, and they are adjacent. So
the conditions both for thinness and proper thinness are
satisfied.

\emph{Case 3: $p < q = r$.} In this case, $(v_q,w_j)$ and
$(v_r,w_{\ell})$ are in different classes, and they are adjacent.
So the conditions both for thinness and proper thinness are
satisfied.

\emph{Case 4: $p < q < r$.} Suppose first that $G_1$ is $k$-thin
(resp. proper $k$-thin) and $(v_p,w_i), (v_q,w_j)$ belong to the
same class, i.e., $i=j$ and $v_p$, $v_q$ belong to the same class
in $G_1$. Since $p < r$, vertices $(v_p,w_i)$ and $(v_r,w_{\ell})$
are adjacent in $G_1 \ltimes G_2$ if and only if $v_p v_r \in E_1$
and $w_i w_{\ell} \not \in E_2$. Since $v_p$, $v_q$ belong to the
same class in $G_1$ and the order and the partition are
consistent, $v_q v_r \in E_1$. Since $i = j$, $w_j w_{\ell} \not
\in E_2$. So $(v_r,w_{\ell})$ and $(v_q,w_j)$ are adjacent in $G_1
\ltimes G_2$, as required.

Now suppose that $G_1$ is proper $k$-thin and $(v_q,w_j),
(v_r,w_{\ell})$ belong to the same class, i.e., $j=\ell$ and
$v_q$, $v_r$ belong to the same class in $G_1$. Since $p < r$,
vertices $(v_p,w_i)$ and $(v_r,w_{\ell})$ are adjacent in $G_1
\ltimes G_2$ if and only if $v_p v_r \in E_1$ and $w_i w_{\ell}
\not \in E_2$. Since $v_q$, $v_r$ belong to the same class in
$G_1$ and the order and the partition are strongly consistent,
$v_p v_q \in E_1$. Since $j=\ell$, $w_i w_j \not \in E_2$. So
$(v_p,w_i)$ and $(v_q,w_j)$ are adjacent in $G_1 \ltimes G_2$, as
required.
\end{proof}

\section{Proofs from the manuscript by Chandran, Mannino and Oriolo}

\noindent \textbf{Theorem~\ref{thm:tK2}.}
\emph{\cite{C-M-O-thinness-man} For every $t \geq 1$,
$\thin(\overline{tK_2})=t$.} \\

\begin{proof} Let $G=\overline{tK_2}$, $t \geq 1$.
Let $V(G) = \{x_1, y_1, \dots, x_t, y_t\}$ and suppose that $(x_i,
y_i)$, for $1 \leq i \leq t$, are the only pairs of non-adjacent
vertices. If we define, for $1 \leq i \leq t$, $V^i = \{x_i,
y_i\}$, then every total order on the vertices of $V(G)$ is
consistent with this partition. We now show that $G$ is not
$(t-1)$-thin. Suppose the contrary, that is, there exist an
ordering $<$ on $V(G)$ and a partition of $V(G)$ into $t-1$
classes $(V^1, \dots, V^{t-1})$ which are consistent. For every
class, denote by $f(V^h)$ the smallest element of $V^h$ with
respect to the ordering $<$. Clearly, there exists at least one
pair $\{x_i, y_i\}$, $1 \leq i \leq t$, such that $\bigcup_h
\{f(V^h)\} \cap \{x_i,y_i\} = \emptyset$. Without loss of
generality, assume that $x_i < y_i$. Let $V^q$ be the class which
$x_i$ belongs to. It follows that $y_i$ is adjacent to $f(V^q)$
and nonadjacent to $x_i$; moreover, $y_i > x_i > f(V^q)$. But this
is a contradiction.
\end{proof}

Recall that the \emph{vertex isoperimetric peak} of a graph $G$,
denoted as $b_v(G)$, is defined as $b_v(G) = \max_s \min_{X\subset
V, |X|=s} |N(X) \cap (V(G) \setminus X)|$, i.e., the maximum over
$s$ of the lower bounds for the number of boundary vertices
(vertices outside
the set with a neighbor in the set) in sets of size $s$.\\

\noindent \textbf{Theorem~\ref{thm:peak}.}
\emph{\cite{C-M-O-thinness-man} For every graph $G$ with at
least one edge,} $\thin(G)\geq b_v(G)/\Delta(G)$. \\

\begin{proof}
Let $t$ be the thinness of $G$, $\Delta = \Delta(G)$, $k = b_v(G)$
and $s$ realizing the vertex isoperimetric peak of $G$. There
exist a partition $V^1, \dots, V^t$ and an ordering $v_1, \dots,
v_n$ of $V(G)$ that are consistent. Let $S = \{v_{n-s+1},\dots,
v_n\}$, that is, $S$ is the set of the $s$ greatest vertices of
$G$. Let $N_S = N(S) \cap (V(G) \setminus S)$. Note that, since
every node in $N_S$ is outside $S$, each of them is smaller than
any vertex in $S$.

We claim that for each $i$, $1 \leq i \leq t$, $|V^i \cap N_S|
\leq \Delta$. (That is, none of the classes can contain more than
$\Delta$ vertices from $N_S$.) Suppose $V^i$ contains at least
$\Delta + 1$ vertices from $N_S$. Let $x$ be the smallest vertex
in $V^i \cap N_S$. Clearly, $x$ is adjacent to some vertex $y \in
S$, since $x \in N_S$. Then, $y$ has to be adjacent to all the
vertices in $V^i \cap N_S$, since all of them are smaller than
$y$. Thus, the degree of $y$ has to be at least $\Delta+1$, a
contradiction. So, each class contains only at most $\Delta$
vertices of $N_S$. It follows that there are at least $|N_S| /
\Delta \geq k / \Delta$ classes.\end{proof}

The following lower bound for the vertex isoperimetric peak of the
grid $G_r$ was proved by Chv\'atalova.

\begin{lemma}\cite{Chv-grid} For every $r$, $b_v(G_r) \geq r$.
\end{lemma}

Now, combining the above results, we get \\

\noindent \textbf{Corollary~\ref{cor:grid}.}
\emph{\cite{C-M-O-thinness-man} For every $r \geq 2$,}
$\thin(GR_r) \geq r/4$. \\

We close with a couple of general upper bounds on the thinness of
a graph. First, we need to introduce some notation and a lemma
whose proof is straightforward.

Let $G$ be a graph. A partition $V^1, \dots, V^k$ of $V(G)$ is
\emph{valid} if there exists an ordering which is consistent with
the partition. A class $V^i$ is called a \emph{singleton class} if
$|V^i| = 1$, otherwise $V^i$ is a \emph{non-singleton class}.

\begin{lemma}\cite{C-M-O-thinness-man}\label{lem:singleton} Let $G$ be a graph with a valid partition $V^1, \dots, V^k$ of $V(G)$. Let $X = \{v \in V : v$ belongs to a singleton
class$\}$. Then there exists an ordering $v_1, \dots v_n$ of
$V(G)$ which is consistent with the partition such that the
vertices of $X$ are the smallest ones in the order.
\end{lemma}

\noindent \textbf{Theorem~\ref{thm:n-log4}.}
\emph{\cite{C-M-O-thinness-man} Let $G$ be a graph. Then}
$\thin(G) \leq |V(G)| - \log(|V(G)|)/4$.
\\

\begin{proof}
Let $t$ be the thinness of $G$, and $n = |V(G)|$. Then there
exists a valid partition of $V(G)$ using $t$ classes. We can
assume that there are at least $2$ singleton classes in this valid
partition. Otherwise, if every class (except possibly one),
contains at least $2$ vertices, then clearly $t \leq (n+1)/2 \leq
n - \log(n)/4$, as required. Let $X = \{u \in V : u$ belongs to a
singleton class$\}$. Let $|X| = h$ and $|V(G) \setminus X| = q$
(clearly, $h + q = n$).

\noindent \textbf{Claim.} $h \leq 2^q$.

Let $V \setminus X = \{y_1, y_2, \dots, y_q\}$. We define a
function $f : X \to \{0, 1\}^q$ as follows: for $x \in X$, $f(x) =
(a_1, a_2, \dots, a_q)$, where $a_i = 1$ if $x$ is adjacent to
$y_i$, and $a_i = 0$ otherwise. Assume for contradiction that $h >
2^q$. Then clearly there exist two vertices $r, s \in X$ such that
$f(r) = f(s)$ (since $|\{0, 1\}^q| = 2^q$). Now, we claim that
even if we merge the two classes containing $r$ and $s$ into one
class, it will remain a valid partition of $V(G)$. This will
provide the required contradiction, since we have assumed that $t$
is the thinness of $G$. By Lemma~\ref{lem:singleton}, there exists
an ordering $\{v_1, \dots, v_n\}$ which is consistent with the
given partition, such that the vertices of the singleton classes
are the smallest, i.e. $X = \{v_1, \dots, v_h\}$. In fact, it is
possible to further assume (without violating the validity of the
partition) that $r = v_{h-1}$ and $s = v_h$. Now, consider merging
the two classes containing $r$ and $s$ into one class. It is not
difficult to verify that the resulting partition is still valid
since the previous ordering is still consistent: there is no
conflict due to vertices of $X \setminus \{r, s\}$, since they are
all smaller than $r$. Also, since $f(r) = f(s)$, it is easy to see
(from the definition of the function $f$) that for any vertex $y
\in V(G) \setminus X$, $y$ is adjacent to $r$ if and only if it is
adjacent to $s$. Therefore, there cannot be any conflict due to
the vertices of $V(G) \setminus X$. Thus, we have a valid
partition of $V(G)$ using only $t-1$ classes, contradicting the
assumption that the thinness of $G$ is $t$. We infer that $h \leq
2^q$. $\diamondsuit$

Now, suppose $h > n - \log(n)/2$. Then, from $2^q \geq h$, we get
$q \geq \log(n - \log(n)/2)$. But this leads to a contradiction
since $n = |X| + |V \setminus X| = h + q \geq n - \log(n)/2 +
\log(n - \log(n)/2) > n$. So, we infer that $h \leq n -
\log(n)/2$. Then, $q = n - h \geq \log(n)/2$. Note that $t \leq h
+ q/2$, since every non-singleton class contains at least $2$
vertices, and there are only a total of $q$ vertices in
non-singleton classes. Thus, $t \leq h + q - q/2 = n - q/2 \leq n
- \log(n)/4$, as required.
\end{proof}

We can also get a bound in terms of the maximum degree of the
graph by modifying the above proof.\\

\noindent \textbf{Theorem~\ref{thm:bound-delta}.}
\emph{\cite{C-M-O-thinness-man} Let $G$ be a graph. Then}
$\thin(G) \leq |V(G)|(\Delta(G)+3)/(\Delta(G)+4)$. \\

\begin{proof}
Let $n = |V(G)|$ and $\Delta = \Delta(G)$. We modify the proof of
Theorem~\ref{thm:n-log4} as follows.

\noindent \textbf{Claim.} $h \leq q(\Delta+2)/2$.

For $x \in X$, we denote by $|f(x)|$ the number of $1$'s in the
$q$-tuple $f(x)$. Clearly, $\sum_{x \in X}|f(x)| \leq \sum_{y \in
V(G)\setminus X} \mbox{degree}(y) \leq q\Delta$. Now partition $X$
into two classes as follows: Let $X_1 = \{x \in X : |f(x)| = 1\}$
and $X_2 = X \setminus X_1$. Note that $|X_1| \leq q$, since
otherwise, if $|X_1|
> q$, there will be two vertices $r, s \in X_1$ such that $f(r) = f(s)$,
leading to a contradiction, as described in the proof of
Theorem~\ref{thm:n-log4}. We can assume that there is no vertex $x
\in X_2$ with $|f(x)| = 0$, since if such a vertex existed we
could have merged the class of this vertex with that of the
greatest vertex in $X$, contradicting the minimality of the size
of the partition. Thus, for every vertex $x \in X_2$, $|f(x)| \geq
2$. Since $\sum_{x\in X_2} |f(x)| \leq \sum_{x \in X} |f(x)| \leq
q\Delta$, we have $|X_2| \leq q\Delta/2$. Thus, $h = |X| = |X_1| +
|X_2| \leq q(\Delta+2)/2$. $\diamondsuit$

It follows that $q = n-h \geq n-q(\Delta+2)/2$. Rearranging, we
get $q \geq 2n/(\Delta +4)$. Now, $t \leq h+q/2 = n-q/2 \leq
n(\Delta+3)/(\Delta+4)$.
\end{proof}

\end{document}